\documentclass{article}
\usepackage{amsmath,amssymb,amsthm}
\usepackage{hyperref}

\newtheorem{theorem}{Theorem}[section]
\newtheorem{lemma}[theorem]{Lemma}
\newtheorem{proposition}[theorem]{Proposition}
\newtheorem{example}[theorem]{Example}
\newtheorem{corollary}[theorem]{Corollary}
\newtheorem{remark}[theorem]{Remark}
\newtheorem{definition}[theorem]{Definition}

\newtheorem{applemma}{Lemma}
\newtheorem{appropo}{Proposition}

\newcommand{\ep}{\hfill $\Box$}
\newcommand{\angX}{\langle X \rangle}

\def\calA{{\mathcal A}}
\def\calB{{\mathcal B}}
\def\calC{{\mathcal C}}
\def\calD{{\mathcal D}}
\def\calE{{\mathcal E}}
\def\calF{{\mathcal F}}
\def\calT{{\mathcal T}}
\def\calU{{\mathcal U}}
\def\calV{{\mathcal V}}

\def\bfK{{\mathbf K}}
\def\bfa{{\mathbf a}}

\def\rmM{{\mathrm M}}
\def\rmT{{\mathrm T}}

\def\vk{\varkappa}
\def\si{\sigma}
\def\om{\omega}
\def\ve{\varepsilon}
\def\vp{\varphi}
\def\Si{\Sigma}
\def\Om{\Omega}
\def\Ga{\Gamma}

\def\etai{\eta^{-1}}
\def\vpi{\varphi^{-1}}
\def\psii{\psi^{-1}}

\def\ra{\rightarrow}
\def\lra{\leftrightarrow}
\def\Ra{\Rightarrow}
\def\LRa{\Leftrightarrow}

\def\es{\emptyset}
\def\se{\subseteq}

\def\Con{\mathrm{Con}}
\def\Eq{\mathrm{Eq}}
\def\FCon{\mathrm{FCon}}
\def\FGCon{\mathrm{FGCon}}
\def\fork{\mathrm{fork}}
\def\GCon{\mathrm{GCon}}
\def\hg{\mathrm{hg}}
\def\ia{\mathrm{ia}}
\def\RA{\mathrm{RA}}
\def\RCon{\mathrm{RCon}}
\def\RGCon{\mathrm{RGCon}}
\def\Reg{\mathbf{Reg}}
\def\root{\mathrm{root}}
\def\rt{\mathrm{rt}}
\def\SA{\mathrm{SA}}
\def\size{\mathrm{size}}
\def\st{\mathrm{st}}
\def\Tr{\mathrm{Tr}}

\def\SalgA{{\mathcal A} = (A,\Sigma )}
\def\SalgB{{\mathcal B} = (B,\Sigma )}
\def\OalgA{{\mathcal A} = (A,\Omega )}
\def\OalgB{{\mathcal B} = (B,\Omega )}
\def\GalgC{{\mathcal C} = (C,\Gamma )}

\def\SX{\Sigma X}
\def\SXt{T_\Sigma(X)}
\def\SXc{C_\Sigma(X)}
\def\OX{\Omega X}
\def\OXt{T_{\Omega}(X)}
\def\OY{\Omega Y}
\def\OYt{T_\Omega(Y)}
\def\OYc{C_\Omega(Y)}
\def\GXt{T_{\Gamma}(X)}
\def\SXta{\calT_\Sigma(X)}
\def\OXta{\calT_\Omega(X)}
\def\OYta{\calT_\Omega(Y)}

\def\varC{\calC = \{\calC(\Si,X)\}}
\def\varU{\calU = \{\calU(\Si,X)\}}
\def\varV{\calV = \{\calV(\Si,X)\}}
\def\RecSX{Rec(\Sigma,X)}
\def\RecOY{Rec(\Omega,Y)}
\def\VRA{\mathbf{VRA}}

\def\VUT{\mathbf{VUT}}
\def\Nil{\mathbf{Nil}}

\title{Varieties of Unranked Tree Languages}

\author{
Magnus Steinby\footnote{Department of
Mathematics and Statistics, University of Turku, FI-20014 Turku,
Finland, email: steinby@utu.fi}
\and
Eija Jurvanen\footnote{Department of Mathematics and Statistics, University of Turku, FI-20014 Turku, Finland, email: jurvanen@utu.fi}
\and
Antonio Cano \footnote{Departamento de Sistemas
Inform\'aticos y Computaci\'on, Universidad Polit\'ecnica de
Valencia, Camino de Vera s/n, P.O. Box: 22012, E-46020 - Valencia,
email:  acano@dsic.upv.es.}
}

\begin{document}

\maketitle

\begin{abstract}
We study varieties that contain unranked tree languages over all alphabets. Trees are labeled with symbols from two alphabets, an unranked operator alphabet and an alphabet used for leaves only. Syntactic algebras of unranked tree languages are defined similarly as for ranked tree languages, and an unranked tree language is shown to be recognizable iff its syntactic algebra is regular, i.e., a finite unranked algebra in which the operations are defined by regular languages over its set of elements. We establish a bijective correspondence between varieties of unranked tree languages and varieties of regular algebras. For this, we develop a basic theory of unranked algebras in which algebras over all operator alphabets are considered together. Finally, we show that the natural unranked counterparts of several general varieties of ranked tree languages form varieties in our sense.

This work parallels closely the theory of general varieties of ranked tree languages and general varieties of finite algebras, but many nontrivial modifications are required. For example, principal varieties as the basic building blocks of varieties of tree languages have to be replaced by what we call quasi-principal varieties, and we device a general scheme for defining these  by certain systems of congruences.
\end{abstract}


\section{Introduction}\label{se:Intro}

In its prevalent form, the theory of tree automata and tree
languages (cf.~\cite{That73}, \cite{GeSt84}, \cite{GeSt97}
or~\cite{TATA07} for general expositions) deals with trees in
which the nodes are labelled with symbols from a ranked alphabet;
in a ranked alphabet each symbol has a unique nonnegative integer
rank, or arity, that specifies the number of children of any node
labelled with that symbol. Thus a ranked alphabet may be viewed as
a finite set of operation symbols in the sense of algebra, and
then trees are conveniently defined as terms and finite
(deterministic, bottom-up) tree automata become essentially finite
algebras. As a matter of fact, the theory of tree automata arose
from the interpretation of ordinary automata as finite unary algebras
advocated by J.R.~B\"uchi and J.B.~Wright already around 1960
(cf.~\cite{GeSt84} or \cite{GeSt97} for notes on this subject and
references to the early literature). Universal algebra has offered
the theory of tree automata a solid foundation, and the definition
of trees as terms links it naturally also with term rewriting.
Nothing of this would be lost even if each symbol in a ranked
alphabet is allowed a fixed finite set of ranks. However,
when trees are used as representations of XML documents or parses
of sentences of a natural language, fixing the possible ranks of a
symbol is awkward. It is in particular the study of XML that
propels the current interest in unranked tree languages
(cf.~\cite{BrMW01}, \cite{Neve02}, \cite{MaNi07}, \cite{TATA07}
or~\cite{Schw07}, for example).

Actually, unranked  trees are nothing new in the theory of tree
languages. Let us note just two early papers, published in 1967
and 1968, respectively. In \cite{That67} Thatcher defines
recognizable unranked tree languages, proves some of their basic
properties, and establishes a connection between them and the
derivation trees of extended context-free grammars; the paper is
motivated by the study of natural languages. Recognizability is
defined using ``pseudoautomata'', a concept attributed to B\"uchi
and Wright, in which state transitions are regulated by regular
languages over the state set. This idea reappears in various forms
in most works on recognizable unranked tree languages, and also in
this paper pseudoautomata play a central role (we call
them regular algebras). In \cite{PaQu68} Pair and Quere consider
hedges (that they call ramifications), i.e., finite sequences of
unranked trees, and they introduce a new class of algebras,
binoids, in terms of which the recognizability of hedge languages
are defined and discussed. Also hedges have become a much used
notion in the theory of unranked tree languages
(cf.~\cite{Taka75,BrMW01,BoWa08,BoSS12}, for example). However, we shall
consider just unranked trees.

The varieties to be studied here contain tree languages over
all unranked alphabets and leaf alphabets, and we take as our
starting point the theory of  general varieties of (ranked) tree
languages presented in \cite{Stei98}. However, in addition to the
modifications to be expected, some novel notions are needed. On
the other hand, the formalism is actually simplified by the fact
that symbols have no ranks.

The paper is organized as follows. In Section 2 we recall some
general preliminaries, introduce unranked trees and some related
notions. In addition to an unranked alphabet, that we call the
operator alphabet, we use also a leaf alphabet for labeling leaves
only. If $\Si$ is an operator alphabet and $X$ a leaf alphabet,
then unranked $\SX$-trees are defined as unranked $\Si$-terms with
variables in $X$, and sets of such terms are called unranked
$\SX$-tree languages. This arrangement with two alphabets will be
convenient for the algebraic treatment of our subject but, as we
shall demonstrate, it is also natural in typical applications.

In Section 3 we develop the basic theory of unranked algebras in a
way that allows us to consider together algebras over different
operator alphabets. Here we can follow quite closely the
corresponding generalized theory of ordinary (i.e., ranked)
algebras as presented in \cite{Stei98}. In the next section we
consider the unranked algebras in terms of which recognizability
is defined. We call them regular algebras, but they are precisely the pseudoautomata
of B\"uchi and Wright mentioned above. We show that the class of
regular algebras is closed under our generalized constructions of
subalgebras, epimorphic images and finite direct products. Thus
they form the greatest variety of regular algebras (VRA). By first
proving a number of commutation and semi-commutation relations
between the class operators corresponding to the various
constructions of algebras, we derive a representation for the VRA
generated by a given class of regular algebras similar to Tarski's
classical HSP-theorem (cf.~\cite{BuSa81} or \cite{Berg12}, for
example). The regular congruences considered in Section 5 are
intimately connected with regular algebras. Indeed, the
congruences of a regular algebra are regular, and the quotient
algebra of an unranked algebra with respect to a congruence is
regular exactly in case the congruence is regular.

In Section 6 we introduce syntactic congruences and syntactic
algebras of subsets of unranked algebras. Also here it is
convenient to consider these notions on this general level.
Syntactic congruences and syntactic algebras of unranked tree
languages are then obtained by viewing them as subsets of term
algebras. All these notions are natural adaptations of their
ranked counterparts (cf. \cite{Stei79,Stei92,Stei05} or
\cite{Alme90}). Our syntactic congruences of unranked tree
languages appear also in \cite{BrMW01} as `top congruences'. In
\cite{BoSS12} the term `syntactic algebra' designates a different
notion that is associated with hedge languages. Similarly as in
\cite{Stei98}, we shall also need reduced syntactic congruences
and algebras obtained by merging symbols that are equivalent with
respect to the subset considered. In Section 7 we define an
unranked tree language to be recognizable if it is recognized by a
regular algebra. This definition is essentially the same as that
of \cite{That67} and equivalent to other definitions that use
finite automata. As one would expect, an unranked tree language is
recognizable if and only if its syntactic algebra is regular, and
the syntactic algebra is in a natural sense the least unranked
algebra recognizing any given unranked tree language. We also show
that the syntactic algebra of any effectively given recognizable
unranked tree language can be effectively constructed; here this
is less obvious than in the ranked case as the operations are
infinite objects and there are infinitely many trees of any given
height $\geq 1$.

In Section 8 we introduce varieties of unranked tree languages
(VUTs). Such a variety contains languages for all operator and
leaf alphabets. Similarly as in the ranked case, a VUT is usually
most naturally defined in terms of congruences of term algebras,
and hence we introduce also varieties of regular congruences
(VRCs) and show that each VRC yields a VUT. Most examples of
varieties of ranked tree languages are so-called principal
varieties or unions of them (cf.~\cite{Stei92,Stei98,Stei05}). A
principal variety corresponds to a variety of congruences that
consists of principal filters of the congruence lattices of term
algebras in which the generating congruences are of finite index.
Because of the unlimited branching  in unranked  trees, the
corresponding VUTs cannot be defined this way. Instead, we
introduce the notion of consistent systems of congruences. These
systems yield varieties of regular congruences and varieties of
unranked tree languages that we call quasi-principal.

In Section 9 we establish a bijective correspondence between the
varieties of regular algebras and the varieties of unranked tree
languages. Section 10 contains several examples of VUTs that may
be regarded as the natural unranked counterparts of some (general)
varieties of ranked tree languages. Thus we have the VUTs of
finite/co-finite, definite, reverse definite, generalized
definite, aperiodic, locally testable and piecewise testable
unranked tree languages. In most cases these VUTs are unions of
quasi-principal VUTs defined by consistent systems of congruences.
For example, for each $k\geq 2$, we define the VUT of $k$-testable
unranked tree languages by means of a consistent system of
congruences that is naturally given by the definition of
$k$-testability, and the VUT of all locally testable unranked
tree languages is the union of these VUTs.

In the final section we briefly review the results of the paper
and note some further topics to be considered.

Several proofs that have been omitted or just outlined in the main text can be found in the Appendix at the end of the paper.

\section{General preliminaries and unranked trees}\label{se:Preliminaries}

We may write $A := B$ to emphasize that $A$ is defined to be $B$.
Similarly, $A :\LRa B$ means that $A$ is defined by the condition
expressed by $B$. For any integer $n\geq 0$, let $[n] :=
\{1,\ldots,n\}$. For a relation $\rho \se A\times B$, the fact
that $(a,b)\in \rho$  is also expressed by $a \,\rho\, b$ or  $a
\equiv_\rho b$. For any $a\in A$, let $a\rho := \{ b \mid a \rho
b\}$, and for $A'\se A$, let $A'\rho := \{b\in B \mid  (\exists
a\in A')\, a\,\rho\, b\}$. The \emph{converse} of $\rho$ is the
relation $\rho^{-1}:=\{ (b,a) \mid a\rho b\}$ ($\se B\times A$).
 The \emph{composition} of two relations
$\rho \se A\times B$ and $\rho' \se B\times C$ is the relation
$\rho \circ \rho' := \{ (a,c) \mid a\in A, c\in C, (\exists b\in
B)\, a \rho b \; \mathrm{and}\; b \rho' c\}$. The set of
equivalence relations on a set $A$ is denoted by $\Eq(A)$, and for
any $\theta\in\Eq(A)$, let $A/\theta := \{a\theta \mid a\in A\}$
be the corresponding quotient set. Let $\Delta_A := \{(a,a)\mid
a\in A\}$ be the \emph{diagonal relation} and $\nabla_A := A\times
A$ be the \emph{universal relation}.

For a mapping $\vp :A \ra B$, the image $\vp(a)$ of an element
$a\in A$ is also denoted by $a\vp$. Accordingly, if $H\se A$ and
$K\se B$, we may also write $H\vp$ and $K\vpi$ for $\vp(H)$ and
$\vpi(K)$, respectively. Especially homomorphisms will be treated
this way as right operators and the composition of $\vp :A \ra B$
and $\psi : B \ra C$ is written as $\vp\psi$. The \emph{identity
map} $A \ra A, a \mapsto a,$ is denoted by $1_A$. For any sets $A_1,\ldots,A_n$ ($n\geq 1$) and any $i\in [n]$, we let $\pi_i$ denote the \emph{$i^{th}$ projection} $A_1\times\ldots\times A_n \ra A_i, \: (a_1,\ldots,a_n) \mapsto a_i$.

For any alphabet $X$, we denote by $X^*$ the set of all words over
$X$, the empty word by $\ve$, and by $Rec(X^*)$ the set of all
regular languages over $X$.

We need also some notions from lattice theory (cf.~\cite{BuSa81}
or \cite{DaPr02}, for example). Let $\leq$ be an order on a set
$L$. A nonempty subset $D$ of $L$ is said to be \emph{directed} if
for all elements $d_1,d_2\in D$, there is an element $d\in D$ such
that $d_1,d_2 \leq d$, and it is a \emph{chain} in $L$ if any two
of its elements are comparable. Of course, any chain is directed.
Now, let $(L,\leq)$ be a lattice. A nonempty subset $F\se L$ is a
\emph{filter} if
\begin{itemize}
  \item[(1)] $a,b\in F$ implies $a\wedge b \in F$, and
  \item[(2)] $a\in F$, $b\in L$ and $a\leq b$ imply $b\in F$.
\end{itemize}
The \emph{filter generated} by a nonempty subset $H$ of $L$, i.e.,
the intersection of all filters containing $H$, is denoted by
$[H)$. It is easy to see that
$$[H) \, =  \, \{b\in L \mid a_1\wedge \ldots \wedge a_n \leq b \; \mathrm{for \; some} \; n \geq 1 \; \mathrm{and} \; a_1,\ldots,a_n \in H\}.$$
As a special case, we get the \emph{principal filter} $[a) :=
\{b\in L \mid a \leq b\}$ generated by a singleton subset
$\{a\}\se L$.

The \emph{unranked trees} to be considered here are finite and
node-labeled, and their branches have a specified left-to-right
order. Both from a theoretical point of view and for some
applications, it will be natural to use two alphabets, an
\emph{operator alphabet} and a \emph{leaf alphabet}, for labelling
our trees. A symbol from the operator alphabet may appear as the
label of any node of a tree, while the symbols of the leaf
alphabet appear as labels of leaves only. In what follows, $\Si$,
$\Om$, $\Ga$ and $\Psi$ denote operator alphabets, and $X$, $Y$
and $Z$ leaf alphabets. We assume that all alphabets are finite
and that operator alphabets are also nonempty.

\begin{definition}\label{de:Trees}{\rm The set $\SXt$ of
\emph{unranked $\SX$-trees} is the smallest set $T$ of strings
such that (1) $X\cup \Sigma \subseteq T$, and
    (2)  $f(t_1,\ldots ,t_m)\in T$ whenever  $f\in
    \Sigma$, $m>0$ and $t_1,\ldots ,t_m\in T$.
Subsets of $\SXt$ are called \emph{unranked $\Sigma X$-tree
languages}. Often we speak simply about \emph{$\SX$-trees} and
\emph{$\Sigma X$-tree languages}, or just about
\emph{(unranked) trees} and \emph{(unranked) tree
languages} without specifying the alphabets. }\ep
\end{definition}

Any $u\in X\cup \Sigma$ represents a one-node tree in which the
only node is labeled with $u$, and $f(t_1,\ldots ,t_m)$ is
interpreted as a tree formed by adjoining the $m$ trees
represented by $t_1, \ldots ,t_m$ to a new $f$-labeled root in
this left-to-right order.

As the definition of the set $\SXt$ is inductive, notions relating
to $\SX$-trees can be defined recursively and statements about
them can be proved by tree induction. For example, the
\emph{height} $\hg(t)$ and the \emph{root (symbol)} $\root(t)$  of
a $\SX$-tree $t$ are defined thus:
\begin{itemize}
  \item[(1)] $\hg(u) = 0$ and $\root(u) =u$  for $u\in \Si \cup X$;
  \item[(2)] $\hg(t) = \max \{\hg(t_1),\ldots ,\hg(t_m)\}+1$ and
  $\root(t) = f$  for $t = f(t_1,\ldots ,t_m)$.
\end{itemize}

\begin{definition}\label{de:Contexts}{\rm
Let $\xi$ be a special symbol not in $\Sigma$ or $X$. A $\Sigma
X$\emph{-context} is a $\Sigma (X\cup \{\xi \})$-tree in which
$\xi$ appears exactly once. Let $\SXc$ denote the set of all
$\Sigma X$-contexts.

If $p,q\in \SXc$, then $p(q)$ is the $\Sigma X$-context obtained
from $p$ by replacing the $\xi$ in it with $q$. Similarly, if
$t\in \SXt$ and $p\in \SXc$, then $p(t)$ is the $\Sigma X$-tree
obtained when the $\xi$ in $p$ is replaced with $t$.
The \emph{height} $\hg(p)$ and the \emph{root} $\root(p)$  of a
$\SX$-context $p$ are defined treating $p$ as a $\Si(X \cup \{\xi\})$-tree.}\ep
\end{definition}

Let us illustrate the above definitions by a few examples.

\begin{example}\label{Ex:Trees&Contexts}{\rm
Let $f,g\in \Sigma$ and $x,y\in X$. The $\SX$-tree $t :=
f(g(y),x,f)$ and the $\SX$-context $p := f(\xi ,f(g))$ are
depicted in Figure \ref{fig:example31}. Now  $\hg(t) = \hg(p) =
2$, $\root(t) = \root(p) = f$, $p(t) = f(f(g(y),x,f) ,f(g))$ and
$p(g(\xi)) = f(g(\xi),f(g))$. On the other hand, $g(\xi)(p) =
g(f(\xi,f(g)))$. }\ep

\begin{figure}[h]
\centering
\setlength{\unitlength}{1mm}
\begin{picture}(28,31)(0,0)
\put(3,3){\line(0,1){8}} 
\put(16,17){\line(0,1){8}} 
\qbezier(3,17)(9.5,21)(16,25) 
\qbezier(29,17)(22.5,21)(16,25) 
\put(2,0){$y$}
\put(2,13){$g$}
\put(15,13){$x$}
\put(28,13){$f$}
\put(15,27){$f$}
\end{picture}
\qquad
\begin{picture}(28,31)(0,0)
\put(25,3){\line(0,1){8}} 
\qbezier(3,17)(8.5,21)(14,25) 
\qbezier(25,17)(19.5,21)(14,25) 
\put(24,0){$g$}
\put(2,13){$\xi$}
\put(24,13){$f$}
\put(13,27){$f$}
\end{picture}
\caption{$\SX$-tree $f(g(y),x,f)$ and $\SX$-context $f(\xi
,f(g))$}
\label{fig:example31}
\end{figure}
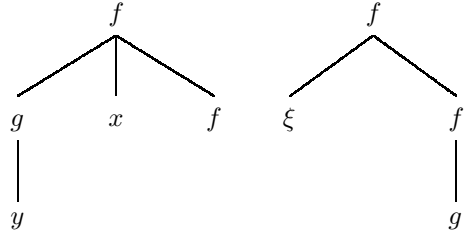
\end{example}

By the following examples we demonstrate that the use of two
alphabets is quite natural also in typical applications of unranked trees.

\begin{example}\label{ex:XML example}{\rm
Figure \ref{fig:XML tree} shows the tree representation of a small
XML document. Here \verb"invoices", \verb"invoice" and \verb"line"
belong to the operator alphabet while \verb"text" is used as a
generic name for a leaf symbol. }\ep

\begin{figure}[h]
\centering
\setlength{\unitlength}{1mm}
\begin{picture}(90,45)(0,0)
\put(2.5,1) {text}
\put(6,5){\line(0,1){8}}
\put(3,14) {line}
\put(40,1) {text}
\put(43.5,5){\line(0,1){8}}
\put(40.5,14) {line}
\put(79.5,1) {text}
\put(83,5){\line(0,1){8}}
\put(80,14) {line}
\put(19,27) {invoice}
\put(77,27) {invoice}
\put(47,40) {invoices}
\qbezier(6,18)(15.5,21.5)(25,25) 
\qbezier(43.5,18)(34.25,21.5)(25,25) 
\put(83,18){\line(0,1){7}} 
\qbezier(25,31)(39.5,35)(54,39) 
\qbezier(83,31)(68.5,35)(54,39) 
\end{picture}
\caption{Unranked tree representing the structure of an XML
document.}
\label{fig:XML tree}
\end{figure}
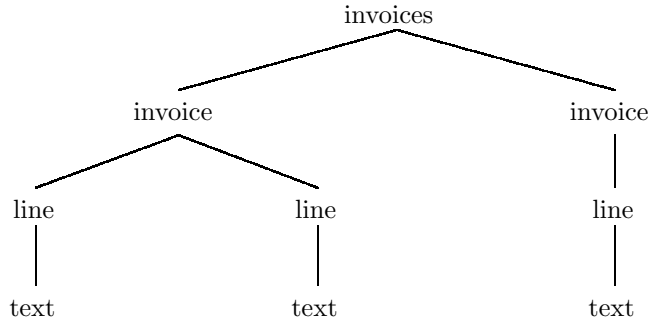
\end{example}

\begin{example}\label{ex:Neven02}{\rm In Example 3.2 of \cite{Neve02}
(also Example 1 of \cite{CrLT05}) the unranked trees are Boolean
expressions without complements and variables, in which
disjunctions and conjunctions may appear with any nonnegative
arities. The alphabet consists of the symbols $\vee, \wedge, 0$
and $1$. In our formalism it is natural to partition this set into
an operator alphabet $\Si = \{\vee,\wedge\}$ and a leaf alphabet
$X = \{0,1\}$; the symbols $0$ and $1$ never label inner nodes.
}\ep
\end{example}

\begin{example}\label{ex:Linguistics}{\rm
Also the parse trees of sentences in a natural language are often
best viewed as unranked. For example, in the tree shown in Figure
\ref{fig:linguistics tree}, the label \verb"NP" has multiple
arities (2 and 3). The symbols S, VP, NP etc. that stand for the
various grammatical categories  form the operator alphabet while
the words ``This'', ``singer'', ``has'' etc. belong to the leaf alphabet.
}\ep

\begin{figure}[h]
\centering
\setlength{\unitlength}{1mm}
\begin{picture}(110,56)(0,0)
\put(2.5,13) {This}
\put(6.5,17){\line(0,1){9}}
\put(3,26.5) {Det}
\put(20,13) {singer}
\put(26,17){\line(0,1){9}}
\put(24.5,26.5) {N}
\qbezier(6.5,30)(11.75,34)(16,38) 
\qbezier(26,30)(21,34)(16,38) 
\put(13,39){NP}
\qbezier(16,42.5)(27,47)(38,51.5) 
\put(37,52){S}
\qbezier(38,51.5)(49.25,47)(60.5,42.5) 
\put(57.5,39){VP}
\qbezier(42.5,30)(51.5,34)(60.5,38) 
\put(41,26.5){V}
\qbezier(78,30)(69.25,34)(60.5,38) 
\put(75,26.5){NP}
\put(42.5,17){\line(0,1){9}}
\put(40.2,13){has}
\qbezier(57,17)(68.5,21.5)(78,26) 
\put(54,13){Det}
\qbezier(77,17)(77.5,21.5)(78,26) 
\put(75.5,13){A}
\qbezier(99,17)(88.5,21.5)(78,26) 
\put(97.5,13){N}
\put(57,4){\line(0,1){8}}
\put(56,0){a}
\put(77,4){\line(0,1){8}}
\put(70.5,0){natural}
\put(99,4){\line(0,1){8}}
\put(93.5,0){talent}
\end{picture}
\caption{A parse tree of a sentence}
\label{fig:linguistics tree}
\end{figure}
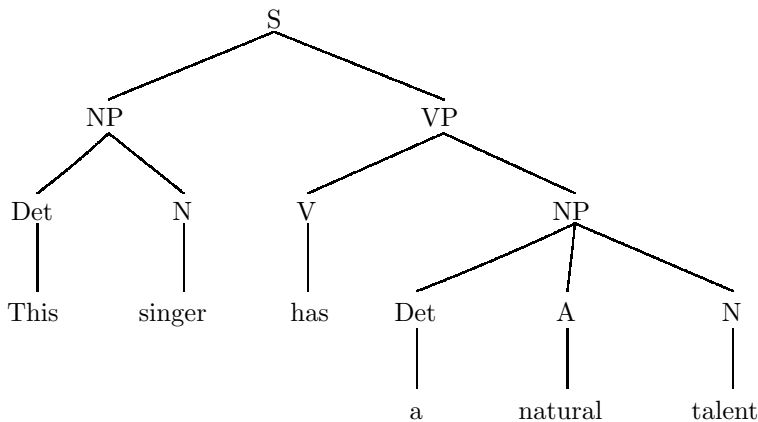

\end{example}

\section{Unranked algebras}\label{se:Unranked algebras}

For an algebraic theory of unranked tree languages we have to
adapt the basic notions and facts of universal algebra to algebras
over unranked sets of operation symbols. Since we will consider
varieties containing tree languages over all alphabets, these
notions are formulated in a  way that corresponds to the
generalized variety theory of \cite{Stei98}. The prefix $g$
appearing in some names stands for ``generalized''.

In the following, the set $A$ of elements of an algebra will be
regarded also as an alphabet and the set of all finite sequences
of elements of $A$ is denoted by $A^*$, and an $m$-tuple
$(a_1,\ldots,a_m)\in A^m$ ($m\geq 0$) may be written as the word
$a_1\ldots a_m$ and subsets of $A^*$  may be viewed as languages.

\begin{definition}\label{de:UAlgebras}{\rm
An \emph{unranked $\Sigma$-algebra} $\calA$ consists of a nonempty
set $A$ (of \emph{elements} of $\calA$) and an operation $f_\calA
: A^* \ra A$ for each $f\in \Si$. We write simply $\SalgA$. The
algebra $\calA$ is \emph{finite} if $A$ is a finite set, and it is
\emph{trivial} if $A$ is a one-element set. We may also speak just
about $\Si$\emph{-algebras} or \emph{algebras} when there
is no danger of confusion.
 }\ep
\end{definition}

These are essentially the "pseudoalgebras" used by Thatcher
\cite{That67} who attributes the concept to J.R.~B\"uchi and
J.B.~Wright (1960). In what follows, $\SalgA$, $\SalgB$, $\OalgB$,
$\GalgC$, etc. are always unranked algebras with the operator
alphabets shown.

\begin{definition}\label{de:Subalgebras}{\rm
If $\Om \se \Si$, an $\Omega$-algebra $\OalgB$ is an
\emph{$\Om$-subalgebra} of a $\Sigma$-algebra $\SalgA$ if  $B\se
A$ and $f_\calB(w) = f_\calA(w)$ for all $f\in \Om$ and $w\in
B^*$. Then we also call $\calB$ a \emph{g-subalgebra} of $\calA$
without specifying $\Om$. If $\Om = \Si$, we call $\calB$ simply a
\emph{subalgebra}. \ep }
\end{definition}

If $\OalgB$ is an $\Om$-subalgebra of $\SalgA$, then $B$ is an
\emph{$\Om$-closed subset} of $\calA$, i.e., $f_\calA(w)\in B$ for
all $f\in\Om$  and $w\in B^*$. Any $\Om$-closed subset $B$ must be
nonempty since  $\Om \neq \es$ and $f_\calA(\ve) \in B$ for every
$f\in \Om$. This means that there is a perfect correspondence
between $\Om$-subalgebras and $\Om$-closed subsets, and hence we
may identify the two.

The intersection of any set of $\Om$-subalgebras is an
$\Om$-subalgebra, and the \emph{$\Om$-subalgebra $\langle H
\rangle_\Om$ generated} by a subset $H\se A$ can be defined in the
usual way as the intersection of all $\Om$-subalgebras that
contain $H$ as a subset. We write $\langle H \rangle_\Si$ simply
as $\langle H \rangle$ and call it the \emph{subalgebra generated}
by $H$.

\begin{definition}\label{de:g-morphism}{\rm
A pair of mappings  $\iota : \Si \ra \Om,\, \vp : A \ra B$ forms a
\emph{g-morphism} from $\SalgA$ to $\OalgB$, written as
$(\iota,\vp) : \calA \ra \calB$, if
$$
    f_{\calA}(a_1,\ldots ,a_m)\vp \: = \:
    \iota(f)_{\calB}(a_1\vp ,\ldots ,a_m\vp)
$$
for all  $f\in \Si$, $m\geq 0$ and $a_1,\ldots,a_m\in A$. A
g-morphism is a \emph{g-epimorphism}, a \emph{g-monomorphism} or a
\emph{g-isomorphism} if both maps are, respectively,  surjective,
injective or bijective. Two algebras $\calA$ and $\calB$ are said
to be \emph{g-isomorphic}, $\calA \cong_g \calB$  in symbols, if
there is a g-isomorphism $(\iota,\varphi) : \calA \ra \calB$, and
$\calB$ is called a \emph{g-image} of $\calA$, if there is a
g-epimorphism from $\calA$ onto $\calB$. Furthermore, an algebra
$\calA$ is said to be \emph{g-covered} by an algebra $\calB$,
$\calA \preceq_g \calB$ in symbols, if $\calA$ is a g-image of a
g-subalgebra of $\calB$. }\ep
\end{definition}

If $(\iota,\vp) : \calA \ra \calB$ is a g-morphism as above, let
$\vp_* : A^* \ra B^*$ be the extension of $\vp$ to a monoid
morphism. Then the fact that $(\iota,\vp)$ is a g-morphism can be
expressed by saying that $f_\calA(w)\vp = \iota(f)_\calB(w\vp_*)$
for all $f\in \Si$ and $w\in A^*$.

A \emph{morphism} $\vp : \calA \ra \calB$ between two
$\Si$-algebras $\SalgA$ and $\SalgB$ is a mapping $\vp : A\ra B$
such that $f_\calA(w)\vp = f_\calB(w\vp_*)$ for all $f\in \Si$ and
$w\in A^*$. It may be viewed as a special g-morphism $(\iota,\vp)
: \calA \ra \calB$ in which $\iota$ is the identity map $1_\Si$.
If $\vp$ is surjective, injective or bijective, then it is called an \emph{epimorphism}, a \emph{monomorphism} or an \emph{isomorphism}, respectively. The algebras $\calA$ and $\calB$ are \emph{isomorphic}, $\calA \cong \calB$ in symbols, if there is an isomorphism $\vp: \calA \to \calB$. Similarly, $\calB$ \emph{covers}  $\calA$,
and we express this by writing $\calA \preceq \calB$, if $\calA$ is
an epimorphic image of some subalgebra of $\calB$.

The g-morphisms of unranked algebras have all the basic properties
of morphisms of ordinary algebras. Some of them are listed in the
following lemma.

\begin{lemma}\label{le:g-morphisms}  Let $\SalgA$, $\OalgB$ and
$\GalgC$ be unranked algebras, and $(\iota,\vp) : \calA \ra \calB$
and $(\vk,\psi) : \calB \ra \calC$ be g-morphisms.
\begin{itemize}
  \item[{\rm (a)}] The product $(\iota\vk,\vp\psi) : \calA \ra
  \calC$ is also a g-morphism. Moreover, if $(\iota,\vp)$ and
  $(\vk,\psi)$ are g-epi-, g-mono- or g-isomorphisms, then so is
  $(\iota\vk,\vp\psi)$.

  \item[{\rm (b)}] If $R$ is a g-subalgebra of $\calB$, then $R\vp^{-1}$ is a g-subalgebra of $\calA$. In particular, if $R$ is a $\Psi$-subalgebra of $\calB$ for some $\Psi \se \Om$, then $R\vp^{-1}$ is a $\iota^{-1}(\Psi)$-subalgebra of $\calA$.

  \item[{\rm (c)}] If $S$ is a g-subalgebra of $\calA$, then $S\vp$ is a g-subalgebra of $\calB$. In particular, if $S$ is a $\Psi$-subalgebra of $\calA$ for some $\Psi\se \Si$, then $S\vp$ is a $\iota(\Psi)$-subalgebra of $\calB$.
\end{itemize}
\end{lemma}

\begin{proof} All of the assertions have straightforward proofs. As an
example we consider statement (b).
Let $R$ be a $\Psi$-subalgebra of $\calB$ for some $\Psi \se \Om$.
To show that $R\vp^{-1}$ is a $\iota^{-1}(\Psi)$-closed subset of
$\calA$, consider  any $f \in \iota^{-1}(\Psi)$, $m\geq 0$ and
$a_1,\ldots,a_m \in R\vp^{-1}$. Since $R$ is $\Psi$-closed,
$\iota(f) \in \Psi$ and $a_1\vp,\dots, a_m\vp \in R$, we get
$$
    f_{\calA} (a_1,\dots,a_m)\vp =
    \iota(f)_{\calB} (a_1\vp,\dots,a_m\vp) \in R,
$$
and hence $f_{\calA} (a_1,\dots,a_m) \in R\vp^{-1}$.
\end{proof}

Next we consider congruences and quotients of unranked algebras.

\begin{definition}\label{de:g-Congruence}{\rm
A \emph{g-congruence} on a $\Si$-algebra $\SalgA$ is a pair
$(\si,\theta)$, where $\si \in \Eq(\Si)$ and $\theta \in \Eq(A)$,
such that for any $f,g\in \Si$, $m \geq 0$  and $a_1,\ldots
,a_m,b_1,\ldots ,b_m\in A$, $$f\, \si \, g, \: a_1 \, \theta \,
b_1, \ldots \, ,a_m \, \theta \, b_m \; \Ra \;
f_{\calA}(a_1,\ldots ,a_m)\: \theta \: g_{\calA}(b_1,\ldots
,b_m).$$ Let $\GCon(\calA)$ denote the set of all g-congruences on
$\calA$.}\ep
\end{definition}

Any algebra $\SalgA$ has as g-congruences at least
$(\Delta_\Si,\Delta_A)$ and $(\si,\nabla_A)$, where  $\si$ is any
equivalence on $\Si$. It is easy to see that  with respect to the
order defined by
$$(\si,\om) \leq (\si',\om') \: :\LRa \: \si \se \si' \:
\mathrm{and} \: \om \se \om' \qquad ((\si,\om),(\si'\om') \in
\GCon(\calA)),$$ $\GCon(\calA)$ forms a complete lattice in which
joins and meets are formed componentwise in $\Eq(\Si)$ and
$\Eq(A)$, respectively.

The ordinary \emph{congruences} of $\SalgA$ are the equivalences
$\theta\in \Eq(A)$ such that $(\Delta_\Si,\theta)\in\GCon(\calA)$.
Their set is denoted by $\Con(\calA)$. Note also that $\theta \in
\Con(\calA)$ whenever $(\om,\theta) \in \GCon(\calA)$ for some
$\om \in \Eq(\Si)$.

\begin{definition}\label{de:QuotientAlgebra}{\rm
For any g-congruence $(\si,\theta) \in \GCon(\calA)$ of an algebra
$\SalgA$, the \emph{g-quotient algebra} $\calA /(\si,\theta) =
(A/\theta ,\Sigma/\si)$ is defined by setting $$(f\si)_{\calA
/(\si,\theta)}(a_1\theta,\ldots ,a_m\theta) = f_{\calA}(a_1,\ldots
,a_m)\theta$$ for all $f\in\Si$, $m\geq 0, $ and $a_1,\ldots
,a_m\in A$.}\ep
\end{definition}

The operations of $\calA/(\si,\theta)$ are clearly well-defined
when $(\si,\theta)\in \GCon(\calA)$. Note also that $(f\si)_{\calA
/(\si,\theta)}(\ve) = f_{\calA}(\ve)\theta$ for every $f\in \Si$.

The following lemma can be proved by appropriately modifying the usual proofs of the corresponding classical statements (cf.~\cite{BuSa81}, for example).

\begin{lemma}\label{le:Quotients&Homomorphisms}
Let $\SalgA$ and $\OalgB$ be any algebras.
\begin{enumerate}
  \item[{\rm (a)}] For any g-congruence $(\si,\theta)$ of $\calA$, the
maps $\theta_\natural : A \ra A/\theta , a \mapsto
a\theta,$ and $\si_\natural : \Sigma \ra \Sigma/\si, f \mapsto
f\si,$ define a g-epimorphism $(\si_\natural,\theta_\natural) :
\calA \ra \calA/(\si,\theta)$.
  \item[{\rm (b)}] The \emph{kernel} $\ker (\iota,\vp) :=
  (\ker \iota, \ker \vp)$ of a g-morphism
  $(\iota,\vp) : \calA \ra \calB$  is a g-congruence of $\calA$, and $\calA/\ker (\iota, \varphi) \cong_g \calB$ if
  $(\iota,\varphi)$ is a g-epimorphism.  \ep
\end{enumerate}
\end{lemma}

The \emph{quotient algebra} $\calA/\theta = (A/\theta,\Si)$ of a
unranked algebra $\SalgA$ with respect to a congruence $\theta \in
\Con(\calA)$ is defined by setting
$f_{\calA/\theta}(a_1\theta,\ldots,a_m\theta) =
f_\calA(a_1,\ldots,a_m)\theta$ for all $f\in \Si$, $m\geq 0$ and
$a_1,\ldots,a_m\in A$. It may be regarded as a special g-quotient
of $\calA$; if we identify in the natural way the operator
alphabets $\Si/\Delta_\Si$ and $\Si$, then $\calA/\theta$ is
isomorphic to $\calA/(\Delta_\Si,\theta)$.

\begin{definition}\label{de:DirectProduct}{\rm
For any mapping $\vk : \Ga \ra \Si \times \Om$, the
\emph{$\vk$-product} of $\SalgA$ and $\OalgB$ is the $\Ga$-algebra
$\vk(\calA,\calB) = (A\times B,\Ga)$ defined so that, for any
$f\in \Gamma$, $m \geq 0$  and $(a_1,b_1),\ldots ,(a_m,b_m)\in
A\times B$,
$$f_{\vk(\calA,\calB)}((a_1,b_1),\ldots ,(a_m,b_m)) \:
= \: (g_\calA(a_1,\ldots ,a_m),h_\calB(b_1,\ldots ,b_m)),$$ where
$(g,h) = \vk(f)$. The products $\vk(\calA_1,\ldots,\calA_n)$ of
any number $n \geq 0$ of unranked algebras are defined similarly.
Without specifying the map $\vk$, we call all such products
jointly \emph{g-products}. For $n = 0$, the g-product
$\vk(\calA_1,\ldots,\calA_n)$ is taken to be the appropriate
trivial algebra (unique up to isomorphism).  }\ep
\end{definition}

The usual \emph{direct product} $\calA_1\times \ldots \times
\calA_n$ of algebras $\calA_1 = (A_1,\Si),\ldots,\calA_n =
(A_n,\Si)$ of the same type may be reconstrued as the g-product
$\vk(\calA_1,\ldots,\calA_n)$ for $\vk: \Si \ra \Si \times \ldots
\times \Si, \, f \mapsto (f,\ldots,f)$. The g-products of just one
factor are of special interest.

\begin{definition}\label{de:Derived algebras}{\rm For any
mapping $\iota : \Si \ra \Om$ and any unranked $\Om$-algebra
$\OalgB$, the $\Si$-algebra $\iota(\calB) = (B,\Si)$ such that
$f_{\iota(\calB)} = \iota(f)_\calB$ for every $f\in \Si$, is
called the \emph{$\iota$-derived algebra} of $\calB$ or, without
specifying $\iota$,  a \emph{g-derived algebra} of $\calB$.
}\ep
\end{definition}

This notion is a natural analog of a special kind of the derived
algebras considered in universal algebra
(cf.~\cite{GrSc90,Schw93}, and also \cite{Stei12}), and it has
similar properties. In particular, we have the following obvious fact.

\begin{lemma}\label{le:HomomDerivAlgebras} If $(\iota,\vp) :
\calA \ra \calB$ is a g-morphism from a $\Si$-algebra $\SalgA$ to
an $\Om$-algebra $\OalgB$, then $\vp : \calA \ra \iota(\calB)$ is
a morphism of $\Si$-algebras. \ep
\end{lemma}


Let us now consider subdirect decompositions of unranked algebras.
Here it suffices to define all the appropriate notions with just
finite algebras in mind. In the following, $\Si_1,\ldots,\Si_n$
and $\Ga$ are operator alphabets, and for the given map  $\vk :
\Ga \ra \Si_1 \times \ldots \times \Si_n$ and each $i\in [n]$, we
denote by $\vk_i$ the composition $\vk\pi_i$ of $\vk$ and the
$i^{th}$ projection $\pi_i : \Si_1 \times \dots \times \Si_n \ra
\Si_i$, i.e., $\vk_i(f) = \vk(f)\pi_i$ for every $f\in \Ga$. Note
that $\pi_i$ also denotes the projection $A_1\times \dots \times
A_n \ra A_i$.

\begin{definition}\label{de:gsd prod and repr}{\rm A
\emph{generalized subdirect product}, a \emph{gsd-product} for
short, of some unranked algebras $\calA_1 =
(A_1,\Si_1),\ldots,\calA_n = (A_n,\Si_n)$ ($n\geq 0$) is a
g-subalgebra $\OalgB$ of a g-product $\vk(\calA_1,\ldots,A_n) =
(A_1\times \dots \times A_n,\Ga)$ such that $B\pi_i = A_i$ and
$\Om\vk_i = \Si_i$ for every $i\in [n]$.

A \emph{gsd-representation} of $\SalgA$, with factors $\calA_1 =
(A_1,\Si_1),\ldots,\calA_n = (A_n,\Si_n)$ ($n\geq 0$), is a
g-monomorphism $(\iota,\vp) : \calA \ra
\vk(\calA_1,\ldots,\calA_n)$, where also $\vk : \Ga \ra \Si_1
\times \dots \times \Si_n$ is injective, such that $A\vp\pi_i =
A_i$ and $\Si\iota\vk_i = \Si_i$ for every $i\in [n]$. Such a
gsd-representation is \emph{proper} if for no $i\in [n]$, both
$\vp\pi_i : A \ra A_i$ and $\iota\vk_i : \Si \ra \Si_i$ are
injective. A finite unranked algebra is \emph{gsd-irreducible} if
it has no proper gsd-representation. }\ep
\end{definition}

That the above gsd-representation  is proper means that none of
the g-epimorphisms $(\iota\vk_i,\vp\pi_i) : \calA \ra \calA_i$ is
a g-isomorphism (cf. Lemma \ref{le:gsd-repr to g-congr} below).

\begin{remark}\label{re:gsd-repr and monom}{\rm If we compose the
maps $\iota : \Si \ra \Ga$ and $\vk : \Ga \ra \Si_1 \times \dots
\times \Si_n$ of the above definition, the gsd-representation
$(\iota,\vp) : \calA \ra \vk(\calA_1,\ldots,\calA_n)$ yields a
monomorphism $\vp : \calA \ra (\iota\vk)(\calA_1,\ldots,\calA_n)$
of $\Si$-algebras. }\ep
\end{remark}

The next two lemmas show the links between the
gsd-representations (with finitely many factors) and the
g-congruences of an unranked algebra. They can be proved in the
same way as their classical counterparts (cf.~\cite{Berg12}, for
example).

\begin{lemma}\label{le:gsd-repr to g-congr} Let $(\iota,\vp) :
\calA \ra \vk(\calA_1,\ldots,\calA_n)$ be a gsd-representation of
an unranked algebra $\SalgA$ with factors $\calA_1 =
(A_1,\Si_1),\ldots,\calA_n = (A_n,\Si_n)$ ($n\geq 0$). For every
$i\in [n]$, $(\iota\vk_i,\vp\pi_i) : \calA \ra \calA_i$ is a
g-epimorphism. Moreover, if we  write $(\si_i,\theta_i) :=
\ker(\iota\vk_i,\vp\pi_i)$ for  each $i\in [n]$, then
\begin{itemize}
  \item[{\rm (a)}] $(\si_i,\theta_i) \in \GCon(\calA)$ and
  $\calA/(\si_i,\theta_i) \cong_g \calA_i$, and
  \item[{\rm (b)}] $(\si_1,\theta_1)\wedge \ldots \wedge (\si_n,\theta_n)
  = (\Delta_\Si, \Delta_A)$.
\end{itemize}
If the representation is proper, then $(\si_i,\theta_i) >
(\Delta_\Si,\Delta_A)$ for every $i\in [n]$. \ep
\end{lemma}

\begin{lemma}\label{le:g-gonr to gsd-repr} If an unranked algebra
$\SalgA$ has g-congruences
$(\si_1,\theta_1)$, \dots, $(\si_n,\theta_n)$ such that
$(\si_1,\theta_1)\wedge \ldots \wedge (\si_n,\theta_n) =
(\Delta_\Si, \Delta_A)$, then $$(1_\Si,\vp) : \calA \ra
\vk(\calA/(\si_1,\theta_1),\ldots,\calA/(\si_n,\theta_n))$$ is a
gsd-representation of $\calA$ for $\vk : \Si \ra \Si/\si_1
\times\dots \times \Si/\si_n, \: f\mapsto (f\si_1,\ldots,f\si_n),$
and $\vp : A \ra A/\theta_1\times \dots \times A/\theta_n, \: a
\mapsto (a\theta_1,\ldots,a\theta_n)$. If $(\si_i,\theta_i) >
(\Delta_\Si,\Delta_A)$ for every $i\in [n]$, then this
gsd-representation is proper.\ep
\end{lemma}

The following proposition contains the counterparts of Birkhoff's
two fundamental theorems about subdirect representations.

\begin{proposition}\label{pr:gsd-irr and gsd-dec} Let $\SalgA$ be a finite unranked algebra.
\begin{itemize}
  \item[{\rm (a)}] $\calA$ is gsd-irreducible if and only if $|A|=|\Si| = 1$ or it has a least nontrivial g-congruence, i.e.,
$\bigcap(\GCon(\calA)\setminus \{(\Delta_\Si,\Delta_A)\}) >
(\Delta_\Si,\Delta_A)$.
  \item[{\rm (b)}]  $\calA$ has a gsd-representation with finitely many factors each of which is a gsd-irreducible g-image of $\calA$.
\end{itemize}
\end{proposition}

\begin{proof} Statement (a) follows from
Lemmas \ref{le:gsd-repr to g-congr} and \ref{le:g-gonr to gsd-repr}. Statement (b) could be
proved similarly as its classical counterpart (cf.~\cite{BuSa81},
for example).
\end{proof}

Note that a trivial $\Si$-algebra $\SalgA$ is gsd-irreducible
exactly in case $|\Si| \leq 2$. Indeed, if $|\Si|>2$, then $\calA$
has a proper gsd-representation in which the factors are trivial
algebras with smaller operator alphabets. If $|\Si|= 2$, then
$(\nabla_\Si,\Delta_A)$ is the least nontrivial g-congruence of
$\calA$.

\begin{definition}\label{de:TermAlgebra}{\rm For any
$\Si$ and $X$, we define the {\em unranked $\SX$-term algebra}
$\calT_\Si(X) = (T_{\Sigma}(X), \Sigma)$ in such a way that for
any $f \in \Si$,
\begin{itemize}
  \item[(1)] $f_{\SXta}(\ve) = f$, and
  \item[(2)] $f_{\calT_\Sigma(X)}(t_1, \ldots , t_m) \ = \ f(t_1,
\ldots , t_m)$ for any $m > 0$ and $t_1, \ldots , t_m \in
T_{\Sigma}(X)$.
\end{itemize}
Again, we may speak simply about the {\em $\SX$-term algebra} or
a \emph{term algebra}. \ep }
\end{definition}

A g-morphism $(\iota,\vp) : \SXta \ra \OYta$ between unranked term
algebras defines a mapping from $\SXt$ to $\OYt$ that replaces
each label $f\in \Si$ with $\iota(f)\in \Om$ and each leaf labeled
with a symbol $x\in X$ with the $\OY$-tree $x\vp$. Such mappings
are the unranked analogs of the \emph{inner alphabetic tree
homomorphisms} considered in \cite{JuPoTh94} (in the form they
appear in \cite{Stei98}). For any given $\iota : \Si \ra \Om$ and
any leaf alphabet $X$, we define a special mapping $\iota_X : \SXt
\ra \OXt$ of this type as follows:
\begin{itemize}
  \item[(1)] $\iota_X(x) = x$ for $x\in X$;
  \item[(2)] $\iota_X(f) = \iota(f)$ for $f\in \Si$;
  \item[(3)] $\iota_X(t) =
  \iota(f)(\iota_X(t_1),\ldots,\iota_X(t_m))$ for $t =
  f(t_1,\ldots,t_m)$.
\end{itemize}
Then $(\iota,\iota_X) : \SXta \ra \OXta$ is a g-morphism that
transforms any $\SX$-tree to an $\OX$-tree by replacing any label
$f\in \Si$ with $\iota(f)$ but preserving all symbols from $X$.

The following proposition expresses an important property of our
term algebras.

\begin{proposition}\label{pr:TermAlgebraFree}
For any $\Si$ and any $X$, the term algebra $\SXta$ is
\emph{freely generated} by $X$ over the class of all unranked
algebras, that is to say,
\begin{itemize}
  \item[{\rm (a)}]  $\angX = \SXt$, and
  \item[{\rm (b)}] if $\OalgA$ is any unranked algebra, then for
  any pair of mappings $\iota : \Si \ra \Om$ and $\alpha : X \ra A$,
  there is a unique g-morphism
  $(\iota,\vp_{\iota,\alpha}) : \SXta \ra \calA$ such that
  $\vp_ {\iota, \alpha} \big|_X  = \alpha$.
\end{itemize}
\end{proposition}

\begin{proof} (a) It is clear that $X\se \angX \se \SXt$, and that to prove
$\SXt \se \angX$, it suffices to show that $\SXt \se B$ for any
$\Si$-closed subset $B$ of $\SXta$ for which $X \se B$. This can
be done simply by tree induction.

(b) For any $\iota:\Sigma\to\Omega$ and $\alpha:X \to A$, a
g-morphism $(\iota,\vp):\SXta \to \calA$ such that $\vp\big|_X =
\alpha$ must satisfy the following conditions:
\begin{itemize}
    \item[(1)] $x\vp = \alpha(x)$ for $x \in X$;
    \item[(2)] $f\vp = f_{\SXta}(\ve)\vp =
          \iota(f)_{\calA}(\ve)$ for $f \in \Si$;
    \item[(3)] $t\vp = \iota(f)_{\calA}(t_1\vp, \dots, t_m\vp)$ for $t= f(t_1,\dots,t_m)$.
\end{itemize}
It is easy to show by tree induction that these conditions assign a unique value $t\vp$ to each $t\in \SXt$, and that the thus defined $(\iota,\vp)$ is a g-morphism.
\end{proof}

Similarly as for ordinary algebras, the values of
$\vp_{\iota,\alpha}$ can be obtained by evaluating term functions
for the valuation $\alpha: X \to A$.

\begin{definition}\label{de:TermFunctions}{\rm For any operator alphabet
$\Si$ and leaf alphabet $X$, the \emph{term function} $t^\calA
\colon A^X \ra A$ of a $\Si$-algebra $\SalgA$ induced by a
$\SX$-tree $t\in \SXt$ is defined as follows. For any $\alpha :
X\ra A$,
\begin{itemize}
  \item[{\rm (1)}] $x^\calA(\alpha) = \alpha(x)$ for any $x\in X$
  \item[{\rm (2)}] $f^\calA(\alpha) = f_\calA(\ve)$ for any
  $f \in \Si$, and
  \item[{\rm (3)}] $t^\calA(\alpha) =
  f_\calA(t_1^\calA(\alpha),\ldots,t_m^\calA(\alpha))$ for $t =
  f(t_1,\ldots,t_m)$. \ep
\end{itemize}
}
\end{definition}

It is easy to see that with $\calA$, $\iota$ and $\alpha$ as in
Proposition~\ref{pr:TermAlgebraFree}, $t\vp_{\iota,\alpha} =
\iota_X(t)^{\calA}(\alpha)$ for every $t \in \SXt$.
The following notion will also be useful.

\begin{definition}\label{de:Translations}{\rm A mapping $p : A \ra A$ is
an \emph{elementary translation} of an unranked $\Si$-algebra
$\SalgA$ if there exist an $f\in \Si$, an $m>0$, and an $i\in [m]$
and elements $a_1$, \dots, $a_{i-1}$, $a_{i+1}$, \dots, $a_m\in A$ such
that $p(b) = f_\calA(a_1\ldots a_{i-1}ba_{i+1}\ldots a_m)$ for all
$b\in A$. The set $\Tr(\calA)$ of \emph{translations} of $\calA$
is defined as the smallest set of mappings $A\ra A$ that contains
the identity map $1_A : A \ra A$ and all elementary translations
of $\calA$, and is closed under composition. }\ep
\end{definition}

The following facts can be shown exactly as for ordinary algebras
(cf.~\cite{BuSa81}, for example).

\begin{lemma}\label{le:Transl and Congr} Any congruence of an unranked algebra $\SalgA$ is invariant with respect to all
translations of $\calA$, and an equivalence on $A$ is a congruence
of $\calA$ if it is invariant with respect to all elementary
translations of $\calA$. \ep
\end{lemma}

Moreover, we have the following counterpart to Lemma 5.3 of
\cite{Stei98}.

\begin{lemma}\label{le:Morphisms and translations} Let
$(\iota,\vp) : \calA \ra \calB$ be a g-morphism from a
$\Si$-algebra $\SalgA$ to an $\Om$-algebra $\OalgB$. For every
translation $p\in \Tr(\calA)$, there is a translation
$p_{\iota,\vp}$ of $\calB$ such that $p(a)\vp =
p_{\iota,\vp}(a\vp)$ for every $a\in A$. If $(\iota,\vp)$ is a
g-epimorphism, then every translation of $\calB$ equals
$p_{\iota,\vp}$ for some $p\in \Tr(\calA)$. \ep
\end{lemma}

Finally, it is easy to see that the translations of a term algebra
$\SXta$ can be defined by and correspond bijectively to the
$\SX$-contexts: for every $p\in \Tr(\SXta)$ there is a unique
$q\in \SXc$ such that $p(t) = q(t)$ for every $t\in \SXt$, and
conversely.

\section{Regular unranked algebras}\label{se:Regular algebras}

Let us now introduce the unranked algebras that will play the same
role here as finite algebras in the ranked case. In \cite{That67}
they were called ``pseudoautomata''.

\begin{definition}\label{de:Regular algebras}{\rm An unranked
algebra $\SalgA$ is said to be \emph{regular} if it is finite and
$f_\calA^{-1}(a)$ is a regular language over $A$ for all $f\in
\Si$ and $a\in A$. \ep }
\end{definition}

The condition that all the sets $f_\calA^{-1}(a) = \{ w\in A^*
\mid f_\calA(w) = a \}$ be regular languages means that the
functions $f_\calA : A^* \ra A$ can be computed by finite
automata.

\begin{example}\label{ex:NonRegular Finite Alg}{\rm
Let $\Sigma =  \{ f\}$ and let $\calA = (\{0 ,1\},\Si)$ be the
unranked $\Si$-algebra such that $f_\calA(w) = 1$ if $ w \in
\{0^n1^n \mid n \geq 0 \}$, and $f_\calA(w) = 0$ otherwise ($w\in
A^*$). In this case ${\cal A}$ is not regular since  the language
$f_\calA^{-1}(1) = \{0^n1^n \mid n \geq 0 \}$ is not regular. On
the other hand, a regular algebra $\calA$ is obtained if we define
$f_\calA$ by setting
$$f_\calA (w) = a_1+\ldots +a_n \mod 2, $$
for $w = a_1\ldots a_n$, where $a_1, \ldots, a_n \in A$. }\ep
\end{example}

\begin{lemma}\label{le:Regular subalgebras and images} All g-subalgebras
and all g-images of a regular algebra are regular.
\end{lemma}

\begin{proof}  Let $\SalgA$ be a regular algebra and let $\OalgB$ be any
unranked algebra.

If $\calB$ is an $\Om$-subalgebra of $\calA$, then
$f_\calB^{-1}(b) = f_\calA^{-1}(b)\cap B^*$ is a regular language
for all $f\in \Om$ and $b\in B$, and hence $\calB$ is regular.

Assume now that there is a g-epimorphism $(\iota,\vp) : \calA \ra
\calB$. Consider any $g\in \Om$, $b\in B$ and $w\in B^*$. If
$\vp_* : A^* \ra B^*$ is the extension of $\vp$ to a monoid
morphism, then $w = v\vp_*$ for some $v\in A^*$. If $f \in \Si$ is
such that $\iota(f) = g$, then
\begin{align*} w \in g_\calB^{-1}(b) \: & \LRa \: g_\calB(w) = b
\: \LRa \: \iota(f)_\calB(v\vp_*) = b \: \LRa f_\calA(v)\vp = b
 \:  \LRa v\in f_\calA^{-1}(b\vp^{-1}) \\
& \Ra w\in f_\calA^{-1}(b\vp^{-1})\vp_*,
\end{align*}
i.e., $g_\calB^{-1}(b) \se f_\calA^{-1}(b\vp^{-1})\vp_*$. For
the converse inclusion, let $w \in f_\calA^{-1}(b\vp^{-1})\vp_*$.
Then $w = v\vp_*$ for some $a\in b\vpi$ and $v\in f_\calA^{-1}(a)$.
This means that
$$b = a\vp = f_\calA(v)\vp = g_\calB(v\vp_*) = g_\calB(w),$$ and
hence $w\in g_\calB^{-1}(b)$. Since $f_\calA^{-1}(b\vp^{-1})$ is
the union of finitely many regular sets $f_\calA^{-1}(a)$, where
$a\in b\vp^{-1}$, this means that $g_\calB^{-1}(b)$ is regular.
\end{proof}

\begin{lemma}\label{le:RegularProduct} Any g-product of regular
algebras is regular. In particular, any g-derived algebra of a
regular algebra is regular.
\end{lemma}

\begin{proof} To simplify the notation, we consider a g-product
$\vk(\calA,\calB) = (A\times B, \Ga)$ of just two regular algebras
$\SalgA$ and $\OalgB$. Let $f\in \Ga$ and $(a,b)\in A\times B$,
and assume that $\vk(f) = (g,h)$. If $\vp_1 : (A\times B)^* \ra
A^*$ and $\vp_2 : (A\times B)^* \ra B^*$ are the extensions of the
projections $\pi_1 : A\times B \ra A$ and $\pi_2 : A\times B \ra
B$ to monoid morphisms, then $f_{\vk(\calA,\calB)}(w) =
(g_\calA(w\vp_1),h_\calB(w\vp_2))$ for any $w\in (A\times B)^*$.
From this it follows that $f_{\vk(\calA,\calB)}^{-1}(a,b) =
g_\calA^{-1}(a)\vp_1^{-1}\cap h_\calB^{-1}(b)\vp_2^{-1}$, which is
a regular set.
\end{proof}

Let us say that a regular algebra $\SalgA$ is \emph{effectively
given} if, for all $f\in \Si$ and $a\in A$, we are given a finite
recognizer of $f_\calA^{-1}(a)$. The following proposition
expresses an important property of regular algebras.

\begin{proposition}\label{pr:TranslRegAlg} For any effectively given regular algebra $\SalgA$, the set $\Tr(\calA)$ of all translations of $\calA$ is effectively computable.
\end{proposition}

\begin{proof}  For any $f\in \Si$ and $u,v\in A^*$, let $f_{u,v} : A \ra A,
a \mapsto f_\calA(uav)$, be the elementary translation defined by
$f$, $u$ and $v$. We show that $E_f := \{f_{u,v} \mid u,v\in
A^*\}$ is effectively computable for each $f$.

For each $a\in A$, we can find a finite monoid $M_a$, a morphism
$\vp_a : A^* \ra M_a$ and a subset $F_a \se M_a$ such that
$f_\calA^{-1}(a) = F_a\vpi_a$. Let $\sim$ be the equivalence on
$A^*$ such that for any $u,v\in A^*$, $u \sim v$ if and only if
$u\vp_a = v\vp_a$ for every $a\in A$. Then $f_{u,v} = f_{u',v'}$
for all words $u,v,u',v'\in A^*$ such that $u \sim u'$ and $v \sim
v'$. Indeed, for all $a,b\in A$,
\begin{align*}
  f_{u,v}(a) = b \: & \LRa \: f_\calA(uav) = b \: \LRa \: uav \in f_\calA^{-1}(b) \: \LRa \: (uav)\vp_b \in F_b\\
  & \LRa \: u\vp_b\cdot a\vp_b \cdot v\vp_b \in F_b \: \LRa \: \: u'\vp_b\cdot a\vp_b \cdot v'\vp_b \in F_b \: \LRa \ldots \\
  & \LRa \: f_{u',v'}(a) = b.
\end{align*}
Let $R$ be a set of representatives of the partition $A^*/\sim$.
Such an $R$ is finite and can be effectively formed using the
regular sets $m\vpi_a$ ($m\in M_a$, $a\in A$). Then $E_f$ is
obtained as the set $\{f_{u,v} \mid u,v\in R\}$.
\end{proof}

We shall consider classes containing  unranked $\Si$-algebras for
any operator alphabet $\Si$. Note that also the operators $S$, $H$
and $P_f$ are now applied to such classes.  The class of all
regular algebras is denoted by $\Reg$.

\begin{definition}\label{de:VRAs}{\rm For any class $\bfK$ of
unranked algebras, let
\begin{itemize}
  \item[(1)]  $S_g(\bfK)$ be the class of algebras g-isomorphic to a g-subalgebra of a member $\bfK$,
  \item[(2)] $H_g(\bfK)$  be the class of all g-images of members of $\bfK$,
  \item[(3)] $P_g(\bfK)$ be the class of algebras isomorphic to g-products of members of $\bfK$,
  \item[(4)] $S(\bfK)$ be the class of algebras isomorphic to a subalgebra of a member $\bfK$,
  \item[(5)] $H(\bfK)$  be the class of all epimorphic images of members of $\bfK$, and
  \item[(6)] $P_f(\bfK)$ be the class of algebras isomorphic to the direct product of a finite family of members of $\bfK$.
\end{itemize}
A class $\bfK$ of regular unranked algebras is a \emph{variety of
regular algebras (VRA)} if $S_g(\bfK)$, $H_g(\bfK)$, $P_g(\bfK) \se
\bfK$. The class of all VRAs is denoted by $\VRA$.}\ep
\end{definition}

Since g-derived algebras are special g-products, the following
fact is an immediate consequence of the definition of VRAs.

\begin{lemma}\label{le:VRAsSolid} Every VRA is closed under
the forming of g-derived algebras. \ep
\end{lemma}

From Lemmas \ref{le:Regular subalgebras and images} and
\ref{le:RegularProduct} we get the following proposition.

\begin{proposition}\label{pr:RegularAlgebrasVFRA} The class
$\Reg$ of all regular unranked algebras is a VRA, and hence it is
the greatest VRA. \ep
\end{proposition}

The intersection of all VRAs that contain a given class $\bfK$ of
regular algebras is obviously a VRA. It is called the \emph{VRA
generated} by $\bfK$ and it is denoted by $V_g(\bfK)$.

If $P$ and $Q$ are any algebra class operators, such as $H_g$,
$S_g$ or $P_f$, we denote, as usual, by $PQ$ the operator such
that for any class $\bfK$ of algebras, $PQ(\bfK) = P(Q(\bfK))$.
Moreover, $P \leq Q$ means that $P(\bfK) \se Q(\bfK)$ for every
class $\bfK$.

The obvious fact that $\bfK \se S(\bfK) \se S_g(\bfK)$, $\bfK \se
H(\bfK) \se H_g(\bfK)$, and $\bfK \se P_f(\bfK) \se P_g(\bfK)$ for
any class  $\bfK$ of finite unranked algebras, will be used
without comment.

To obtain an analog of Tarski's HSP-representation of the
generated variety-operator (cf.~\cite{BuSa81} or \cite{Berg12})
for $V_g$, we first prove some commutation and semi-commutation
relations of our class operators.

\begin{lemma}\label{le:SHPcommut}
\begin{itemize}
    \item[{\rm (a)}]  $S_gS_g = S_gS = SS_g = S_g$.
    \item[{\rm (b)}]  $H_gH_g = H_gH = HH_g = H_g$.
    \item[{\rm (c)}]  $P_gP_g = P_gP_f = P_fP_g = P_g$.
    \item[{\rm (d)}]  $S_gH \leq S_gH_g \leq HS_g \leq H_gS_g$.
    \item[{\rm (e)}]  $P_gS \leq P_gS_g \leq SP_g \leq S_gP_g$.
    \item[{\rm (f)}]  $P_gH \leq P_gH_g \leq HP_g \leq H_gP_g$
\end{itemize}
\end{lemma}

\begin{proof} Statements (a) and (b) hold because obviously $S_gS_g =
S_g$ and $H_gH_g = H_g$.
To prove (c), it clearly suffices to show that $P_gP_g \leq P_g$,
and in each of (d), (e) and (f), it suffices to prove the second
inequality because the other two are obvious. By the way of an
example, we verify the inequality $S_gH_g \le HS_g$.

Let $\bfK$ be any class of unranked algebras. To construct
a typical member  $\GalgC$ of $S_gH_g(\bfK)$,
let $\calA = (A,\Sigma)$ be in $\bfK$, $(\iota,\vp): \calA \ra
\calA'$ be a g-epimorphism, $\OalgB$ be a g-subalgebra of
$\calA'$, and let $(\vk,\psi): \calB \to \calC$ be a
g-isomorphism. Now $\calB\vp^{-1} =
(B\vp^{-1},\iota^{-1}(\Omega))$ is a g-subalgebra of $\calA$. If
we choose a subset $\Sigma'$ of $\iota^{-1}(\Omega)$ so that the
restriction of $\iota$ to $\Sigma'$ is a bijection $\iota':
\Sigma' \to \Omega$, then $\calD = (B\vp^{-1},\Sigma')$ is a
g-subalgebra of $\calA$.

Next we define a $\Gamma$-algebra $\calE = (B\vp^{-1},\Gamma)$ so
that for each $g \in \Gamma$, $g_{\calE} = f_{\calD}$ for the $f
\in \Sigma'$ such that $g = \vk(\iota'(f))$. Then $(\iota'\vk,
1_{B\vp^{-1}}) : \calD \ra \calE$ is a g-isomorphism. Indeed, if
$f \in \Sigma'$ and $w \in (B\vp^{-1})^*$, then
$$
    f_{\calD} (w) 1_{B\vp^{-1}}  = f_{\calD} (w)
    = \vk(\iota'(f))_{\calE} (w 1_{B\vp^{-1}}).
$$
This means that $\calE \in S_g(\bfK)$. Next, we show that
$\vp\psi: \calE \to \calC$ is an epimorphism. Clearly,
$B\vp^{-1}\vp\psi = C$. Consider any $g \in \Gamma$ and $w \in
(B\vp^{-1})^*$. Let $f \in \Sigma'$ and $h \in \Omega$ be such
that $\iota'(f) = h$ and $\vk(h) = g$. Then
$$
    g_{\calE} (w) \vp\psi = f_{\calD} (w) \vp \psi
    = f_{\calA} (w) \vp \psi = h_{\calA'} (w\vp) \psi
    = h_{\calB} (w\vp) \psi = g_{\calC} (w\vp\psi).
$$
Thus $\calC \in HS_g(\bfK)$.
\end{proof}

The relations $S_gS_g = S_g$, $H_gH_g = H_g$, $P_gP_g =  P_g$,
$S_gH_g \leq H_gS_g$, $P_gS_g \leq S_gP_g$, and $P_gH_g \leq
H_gP_g$ yield, in the usual way (cf.~\cite{BuSa81} or~\cite{Berg12},
for example) the following representation of the $V_g$-operator.

\begin{proposition}\label{pr:V_g = H_gS_gP_g} $V_g = H_gS_gP_g$. \ep
\end{proposition}

For a simpler representation of the $V_g$-operator, in which
just the $P$-operator appears in the generalized form, we need
also the following two relations.

\begin{lemma}\label{le:H_gS_P_g relations}
{\rm (a)}  $H_gS \le HS_g$,  and {\rm (b)} $S_gP_g \leq SP_g$.
\end{lemma}

\begin{proof} Throughout this proof, $\bfK$ is again any given class of
unranked  algebras.

For proving (a), let $\calA = (A,\Sigma) \in \bfK$,
$\calB = (B,\Sigma)$ be a subalgebra of $\calA$,
and let $(\iota,\vp): \calB \ra \calC$ be a g-epimorphism onto a $\Gamma$-algebra $\calC = (C,\Ga)$.

Let $\Omega \se \Sigma$ be such that the restriction $\iota':
\Omega \ra \Gamma$ of $\iota$ to $\Omega$ is a bijection. Then
$\calB' = (B, \Omega)$ is a g-subalgebra of $\calA$.
Define a new algebra $\calB'' = (B,\Gamma)$ by setting for each $g \in \Gamma$,
$g_{\calB''} = h_{\calB'}$ for the $h \in \Omega$ such that
$\iota'(h) = g$. It is easy to show that $(\iota',1_B) : \calB'
\ra \calB''$ is a g-isomorphism. Hence, $\calB'' \in
S_g(\bfK)$. To prove $\calC \in HS_g(\bfK)$, it suffices to
verify that  $\vp: \calB'' \to \calC$ is an epimorphism of
$\Ga$-algebras.  Consider any $g \in \Gamma$, $m \ge 0$, and
$b_1$, \dots, $b_m \in B$, and let  $h \in \Omega$ be the symbol
such that $\iota(h) = g$. Then
$$
    g_{\calB''} (b_1, \dots, b_m) \vp = h_{\calB'} (b_1, \dots, b_m) \vp
    = h_{\calB} (b_1, \dots, b_m) \vp
    = g_{\calC} (b_1\vp, \dots, b_m\vp).
$$
Since $\vp$ is surjective, this shows that $\vp$ is an
epimorphism.

To prove (b), let $n\geq 0$, $\calA_i = (A_i,\Sigma_i) \in \bfK$
for each $i\in [n]$,  $\vk: \Omega \to \Sigma_1 \times \dots
\times \Sigma_n$ be a mapping, $\psi: \vk(\calA_1, \dots, \calA_n)
\ra \calB$ be an isomorphism to an algebra $\OalgB$, $\calC =
(C,\Omega')$ be a g-subalgebra of $\calB$,  and $(\iota,\vp)$ a
g-isomorphism from $\calC$ to $\calD = (D,\Gamma)$. Then $\calD$
is a typical representative of $S_gP_g(\bfK)$.

Let $\lambda: \Gamma \to \Sigma_1 \times \dots \times \Sigma_n$ be
the mapping such that $\lambda(g) = \vk(\iota^{-1}(g))$ for each
$g\in \Ga$. Then $\calE = (C\psi^{-1},\Gamma)$ is a subalgebra of
the g-product $\lambda(\calA_1, \dots, \calA_n) =
(A_1\times\ldots\times A_n,\Ga)$. Indeed, let $g \in \Gamma$, $m
\ge 0$, and $\bfa_1 = (a_{11}, \dots, a_{1n})$, \dots, $\bfa_m =
(a_{m1}, \dots, a_{mn}) \in C\psi^{-1}$. If $\iota^{-1}(g) = h \in
\Omega'$ and $\vk(h) = (f_1, \dots, f_n)$, then
\begin{align*}
    g_{\lambda(\calA_1, \dots, \calA_n)}(\bfa_1, \dots, \bfa_m) \psi
     = ( (f_1)_{\calA_1} (a_{11}, \dots, a_{m1}), \dots,
         (f_n)_{\calA_n} (a_{1n}, \dots, a_{mn}) ) \psi \\
     = h_{\vk(\calA_1, \dots, \calA_n)}
      (\bfa_1, \dots, \bfa_m) \psi
     = h_{\calB} (\bfa_1\psi, \dots, \bfa_m\psi)
     = h_{\calC} (\bfa_1\psi, \dots, \bfa_m\psi)
\end{align*}
is in $C$ since $C$ is an $\Omega'$-closed subset of $\calB$.
To show that also $\calD$ is in $SP_g(\bfK)$, we verify that
$\psi\vp: \calE \ra \calD$, with $\psi$ restricted to $C\vpi$, is an isomorphism. For this, consider
any $g \in \Gamma$, $m \ge 0$ and $\bfa_1$, \dots, $\bfa_m \in
C\psi^{-1}$. If $h\in \Om'$ is the symbol such that $\iota(h) =
g$, then
\begin{align*}
    g_{\calE}
      (\bfa_1, \dots, \bfa_m) \psi\vp &= g_{\lambda(\calA_1, \dots, \calA_n)}
         (\bfa_1, \dots, \bfa_m) \psi\vp
    = h_{\calC} (\bfa_1\psi, \dots, \bfa_m\psi) \vp \\
    &= g_{\calD} (\bfa_1\psi\vp, \dots, \bfa_m\psi\vp).
\end{align*}
Moreover, it is clear that $\psi\vp$ is bijective.
\end{proof}

Now we get the simplified representation of the $V_g$-operator.

\begin{proposition}\label{pr:V_g = HSP_g} $V_g = HSP_g$.
\end{proposition}

\begin{proof}  Since $\bfK \se HSP_g(\bfK) \se H_gS_gP_g(\bfK) = V_g(\bfK)$
for any class $\bfK$ of regular algebras, it suffices to verify
that $HSP_g(\bfK)$ is a VRA. This holds because
\begin{align*}
    S_g(HSP_g) &\le HS_gSP_g \le HS_gP_g \le HSP_g, \\
    H_g(HSP_g) &\le H_gSP_g \le HS_gP_g \le HSP_g, \; \mathrm{and}\\
    P_g(HSP_g) &\le HP_gSP_g \le HSP_gP_g = HSP_g,
\end{align*}
by Lemmas~\ref{le:SHPcommut} and~\ref{le:H_gS_P_g relations}.
\end{proof}

Finally, let us note the following important fact.

\begin{lemma}\label{le:SAs, RAs and VRAs} Let $\bfK$ be a VRA. If $(\si,\theta)$ is a g-congruence of an unranked algebra $\SalgA$, then $\calA/\theta \in \bfK$ iff $\calA/(\si,\theta) \in \bfK$.
\end{lemma}

\begin{proof} It is easy to verify that $(\si_\natural,1_{A/\theta}) :
\calA/\theta \ra \calA/(\si,\theta)$ is a g-epimorphism. Hence,
$\calA/\theta \in \bfK$ implies $\calA/(\si,\theta) \in \bfK$.

Assume now that $\calA/(\si,\theta) \in \bfK$. The g-derived
algebra $\si_\natural(\calA/(\si,\theta))$ is actually the algebra
$\calA/\theta$. Indeed, both are $\Si$-algebras with the same set
$A/\theta$ of elements, and for any $f\in \Si$, $m\geq 0$ and
$a_1,\ldots,a_m\in A$,
\begin{align*}
f_{\si_\natural(\calA/(\si,\theta))}(a_1\theta,\ldots,a_m\theta)
\: &= \: (f\si)_{\calA/(\si,\theta)}(a_1\theta,\ldots,a_m\theta)\:
= \: f_\calA(a_1,\ldots,a_m)\theta\\
&= \: f_{\calA/\theta}(a_1\theta,\ldots,a_m\theta).
\end{align*}
Hence, $\calA/\theta \in \bfK$ by Lemma \ref{le:VRAsSolid}.
\end{proof}

\section{Regular congruences}\label{se:Reg congr}
We shall now introduce the congruences that play here the same
role as the congruences of finite index in the theory of ranked
varieties. In what follows, $\SalgA$ is again  any unranked
algebra. Let $\FCon(\calA) := \{\theta \in \Con(\calA) \mid
A/\theta \; \mathrm{finite}\}$ and let $\FGCon(\calA) :=
\{(\si,\theta)\in \GCon(\calA) \mid \theta \in \FCon(\calA)\}$. If
$\theta\in \Eq(A)$, we treat $A/\theta$ also as an alphabet
and $(A/\theta)^*$ as the free monoid generated by it.

\begin{definition}\label{de:Regular congr}{\rm We call a congruence
$\theta$ of $\calA$ \emph{regular} if $\theta \in \FCon(\calA)$
and for every $f\in \Si$ and every $a\in A$,
$f_{\calA/\theta}^{-1}(a\theta)$ is a regular language over
$A/\theta$. A g-congruence $(\si,\theta)$ of $\calA$ is
\emph{regular} if $\theta$ is a regular congruence. Let
$\RCon(\calA)$ and $\RGCon(\calA)$, respectively, denote the sets
of regular congruences and regular g-congruences of $\calA$. }\ep
\end{definition}

For any $\theta\in \Con(\calA)$, let $\eta_\theta : A^* \ra
(A/\theta)^*$ be the monoid morphism such that $a\eta_\theta =
a\theta$ for each $a\in A$. It is easy to see that for all $f\in
\Si$, $a\in A$ and $w\in A^*$, $w\eta_\theta \in
f_{\calA/\theta}^{-1}(a\theta)$ iff $f_\calA(w) \in a\theta$, and
hence $f_\calA^{-1}(a\theta) =
f_{\calA/\theta}^{-1}(a\theta)\eta_\theta^{-1}$. Since
$\eta_\theta$ is surjective, this implies also
$f_{\calA/\theta}^{-1}(a\theta) =
f_\calA^{-1}(a\theta)\eta_\theta$. These equalities yield the
following lemma that can be used for showing that a congruence is
regular. Note that in $f_{\calA/\theta}^{-1}(a\theta)$, the set
$a\theta$ is regarded as an element of $A/\theta$, but in
$f_\calA^{-1}(a\theta)$ it is a subset of $A$.

\begin{lemma}\label{le:Congr regular} Let $\theta \in \Con(\calA)$.
For any $f\in \Si$ and $a\in A$, $f_{\calA/\theta}^{-1}(a\theta)$
is a regular language over $A/\theta$ iff $f_\calA^{-1}(a\theta)$
is a regular language over $A$. \ep
\end{lemma}

Let us consider a simple example of a finite non-regular
congruence.

\begin{example}\label{ex:Non-regular congruence}{\rm
 For $\Sigma = \{f\}$ and $X=\es$, let $T$ be the $\SX$-tree language such that
\begin{itemize}
  \item[(1)] $f\in T$, and
  \item[(2)] for any $n>0$ and $t_1,\ldots,t_n \in \SXt$, $f(t_1,\ldots,t_n) \in T$ iff (a) $n$ is even, (b) $t_1 ,\ldots , t_{n/2} \in T$, and (c) $t_{n/2+1} ,\ldots , t_{n} \notin T$.
\end{itemize}
It is easy to see that $\theta \in \FCon(\SXta)$ when $\theta$ is
defined by $ s \theta t :\LRa ( s \in T \lra t \in T )$ ($s,t \in
\SXt$). However, $f_{\SXta/\theta}^{-1}(f\theta) = \{ (f \theta)^n
(f(f) \theta)^n \mid n \geq 0 \}$ is not a regular language, and
hence $\theta$ is not a regular congruence.}\ep
\end{example}

If $\SalgA$ is a regular algebra, then for any $\theta \in
\Con(\calA)$,  $f\in \Si$ and $a\in A$,  $f_\calA^{-1}(a\theta)$
is the union of the finitely many regular sets $f_\calA^{-1}(b)$,
where $b\in a\theta$. Hence, Lemma \ref{le:Congr regular} yields
the following proposition.

\begin{proposition}\label{pr:RegCongr of RegAlg} Every congruence of a regular algebra is regular. \ep
\end{proposition}

It is clear that $\FCon(\calA)$ and $\FGCon(\calA)$ are filters of
the lattices $\Con(\calA)$ and $\GCon(\calA)$, respectively.
Moreover, the following hold.

\begin{lemma}\label{le:FCon filter} For any unranked algebra
$\SalgA$, $\RCon(\calA)$ is a filter of the lattice $\Con(\calA)$,
and similarly, $\RGCon(\calA)$ is a filter of $\GCon(\calA)$.
\end{lemma}

\begin{proof} Since $\RCon(\calA)$ contains at least $\nabla_A$, it is
nonempty.

If $\theta,\rho\in \RCon(\calA)$, then clearly $\theta\cap \rho
\in \FCon(\calA)$. Moreover, for any $f\in \Si$ and $a\in A$,
$f_\calA^{-1}(a(\theta\cap\rho)) = f_\calA^{-1}(a\theta) \cap
f_\calA^{-1}(a\rho)$, and hence also
$f_{\calA/\theta\cap\rho}^{-1}(a(\theta\cap\rho))$ is regular by
Lemma \ref{le:Congr regular}.

Next, let $\theta\in \RCon(\calA)$, $\rho\in \Con(\calA)$
and $\theta \se \rho$. Of course, $\rho \in \FCon(\calA)$.
Moreover, for each $a\in A$ there is a finite set of elements
$a_1,\ldots,a_k\in A$ ($k\geq 1$) such that $a\rho = a_1\theta
\cup \ldots \cup a_k\theta$, and hence $f_\calA^{-1}(a\rho) =
f_\calA^{-1}(a_1\theta) \cup \ldots \cup f_\calA^{-1}(a_k\theta)$
is a regular language for every $f\in \Si$. Hence $\rho$ is
regular by Lemma \ref{le:Congr regular}.

That $\RGCon(\calA)$ is a filter of $\GCon(\calA)$ follows
immediately from the fact that $\RCon(\calA)$ is a filter of
$\Con(\calA)$.
\end{proof}

The following connection between regular algebras and regular
congruences is a direct consequence of the definitions of these
concepts.

\begin{proposition}\label{pr:RegCon and RegAlg} If $\theta$ is
a congruence of an unranked algebra $\calA$, then $\calA/\theta$
is a regular algebra exactly in case $\theta$ is a regular
congruence. Similarly, for any g-congruence $(\si,\theta)$ of
$\calA$, the g-quotient  $\calA/(\si,\theta)$ is regular if and
only if $(\si,\theta)\in \RGCon(\calA)$.\ep
\end{proposition}

\section{Syntactic congruences and algebras}\label{SA and RA}

The syntactic congruences and the syntactic algebras of unranked
tree languages are defined in quite the same way   as for ranked
tree languages. Also here it is advantageous to define the notions
for subsets of general unranked algebras.

\begin{definition}\label{de:Syntactic Congr and Alg}{\rm The \emph{syntactic congruence}  $\theta_H$ of a subset $H\se A$ of an unranked algebra $\SalgA$ is defined by
$$
    a\, \theta_H \, b \: :\LRa \: (\forall p\in \Tr(\calA))(p(a)\in H \lra p(b)\in H) \qquad (a,b \in A),
$$
and $\SA(H) := \calA/\theta_H$ is the \emph{syntactic algebra} of
$H$. The natural morphism $\vp_H : \calA \ra \SA(H), \, a \mapsto
a\theta_H,$ is called the \emph{syntactic morphism} of $H$. }\ep
\end{definition}

Let us recall that an equivalence $\theta \in \Eq(A)$ is said to
\emph{saturate} a subset $H\se A$ if $H$ is the union of some
$\theta$-classes.
The following lemma can be proved exactly in the same way as for string languages or ordinary tree languages.

\begin{lemma}\label{le:SyntCongrSat} For any subset $H\se A$ of an unranked algebra $\SalgA$, $\theta_H$ is the greatest congruence of $\calA$ that saturates $H$. \ep
\end{lemma}

The following fact is an immediate consequence of Proposition
\ref{pr:TranslRegAlg}.

\begin{proposition}\label{pr:SA(H) effectively computed} If $\SalgA$ is an effectively given regular algebra, then the syntactic congruence $\theta_H$ and the syntactic algebra $\SA(H)$ of any subset $H \se A$ can be effectively constructed. \ep
\end{proposition}

\begin{definition}\label{de:Syntactic Algebra}{\rm An unranked algebra is called \emph{syntactic} if it is isomorphic to the syntactic algebra of a subset of some unranked algebra. A subset $D\se A$ of an unranked algebra $\SalgA$ is \emph{disjunctive} if $\theta_D = \Delta_A$.\ep
}
\end{definition}

The following facts can be proved exactly as their well-known
counterparts in the ranked case (cf.~\cite{Stei79,Stei92,
Stei05}).

\begin{proposition}\label{pr:SyntAlg and DisjSubs SdiAlg}
 An unranked algebra is syntactic iff it has a disjunctive subset. \ep
\end{proposition}

\begin{proposition}\label{pr:SDS-irr is syntactic} Every finite gsd-irreducible unranked algebra is syntactic. Hence every VRA is generated by regular syntactic algebras. \ep
\end{proposition}

It is easy to see that for any congruence $\theta$ of an unranked
algebra $\SalgA$, there is a greatest equivalence $\rmM(\theta)$
on $\Si$ such that $(\rmM(\theta),\theta) \in \GCon(\calA)$. We
shall need the following properties of the $\rmM$-operator (cf. Lemma~6.1 of~\cite{Stei98}).

\begin{lemma}\label{le:M-operator} Let $\SalgA$ and $\OalgB$ be unranked algebras.
\begin{itemize}
  \item[{\rm (a)}] If $\theta$, $\theta'\in \Con(\calA)$ and $\theta \se \theta'$, then $\rmM(\theta) \se \rmM(\theta')$.
  \item[{\rm (b)}] For any set $\{\theta_i \mid i \in I\}$ of congruences of $\calA$, $\rmM(\bigcap_{i\in I}\theta_i) = \bigcap_{i\in I}\rmM(\theta_i)$.
  \item[{\rm (c)}] For any g-morphism $(\iota,\vp) : \calA \ra \calB$  and any $\theta \in \Con(\calB)$, $$\iota\circ\rmM(\theta)\circ\iota^{-1} \se \rmM(\vp\circ\theta\circ\vpi).$$
      If $\vp$ is surjective, then equality holds. \ep
\end{itemize}
\end{lemma}

Although we shall mainly use the syntactic congruences and
algebras defined above, the following input-reduced versions will
be needed for the proof of our variety theorem.

\begin{definition}\label{de:RedSyntAlg}{\rm The \emph{reduced syntactic congruence}
of a subset $H$ of an unranked algebra $\SalgA$ is the g-congruence
$(\si_H,\theta_H)$ of $\calA$, where $\theta_H$ is the syntactic congruence
of $H$ and $\si_H := \rmM(\theta_H)$, the \emph{reduced syntactic algebra}
$\RA(H)$ of $H$ is the g-quotient $\calA/(\si_H,\theta_H) = (A/\theta_H,\Si/\si_H)$,
and the \emph{syntactic g-morphism} $(\iota_H,\vp_H) : \calA \ra \RA(H)$ is
defined by $\iota_H : f \mapsto f\si_H$ and $\vp_H : a \ra a\theta_H$.}\ep
\end{definition}

The following proposition corresponds to Proposition 5.4 of
\cite{Stei98}.

\begin{proposition}\label{pr:SyntCongr and Operations} Let $\SalgA$ and
$\OalgB$ be unranked algebras.
\begin{itemize}
  \item[{\rm (a)}] $\theta_{A\setminus H} = \theta_H$ for every $H \se A$.
  \item[{\rm (b)}] $\theta_H\cap\theta_K \se \theta_{H\cap K}$ for all $H,K\se A$.
  \item[{\rm (c)}] $\theta_H \se \theta_{p^{-1}(H)}$ for all $H \se A$
  and $p\in \Tr(\calA)$.
  \item[{\rm (d)}] If $(\iota,\vp) : \calA \ra \calB$ is a g-morphism
  and $H\se B$, then $\vp\circ\theta_H\circ\vpi \se \theta_{H\vpi}$ and
  $\iota \circ \si_H \circ \iota^{-1} \se \si_{H\vpi}$, and equalities
  hold if $(\iota,\vp)$ is a g-epimorphism.
\end{itemize}
\end{proposition}

\begin{proof} Assertions (a)--(c) follow directly from the definition of
syntactic congruences, so we prove just (d). For $a$, $a' \in A$,
\begin{align*}
  a \,\vp\circ\theta_H\circ\vpi\, a' \: &\LRa \:
    (\forall q\in \Tr(\calB))\, (q(a\vp)\in H \lra q(a'\vp)\in H) \\
    &\Ra \: (\forall p\in \Tr(\calA))\,
        (p_{\iota,\vp}(a\vp)\in H \lra p_{\iota,\vp}(a'\vp)\in H) \\
    &\LRa \: (\forall p\in \Tr(\calA))\,(p(a)\in H\vpi \lra p(a')\in H\vpi)\\
    &\LRa \: a \, \theta_{H\vpi} \, a'.
\end{align*}
This proves the first inclusion of (d), and the second one now
follows from Lemma~\ref{le:M-operator}:
$$
    \iota \circ \si_H \circ \iota^{-1} =
    \iota \circ \rmM(\theta_H) \circ \iota^{-1} \se
    \rmM(\vp \circ \theta_H \circ \vpi) \se \rmM(\theta_{H\vpi})
    = \si_{H\vpi}.
$$
If $(\iota,\vp)$ is a g-epimorphism, the only ``$\Ra$'' in the
proof of the first inclusion can be replaced by ``$\LRa$'', and
all inclusions become equalities.
\end{proof}

\begin{proposition}\label{pr:SyntAlg and Operations} Let $\SalgA$ and $\OalgB$ be unranked algebras.
\begin{itemize}
  \item[{\rm (a)}] $\SA(A\setminus H) = \SA(H)$ for every $H \se A$.
  \item[{\rm (b)}] $\SA(H\cap K) \preceq \SA(H)\times \SA(K)$ for all $H,K\se A$.
  \item[{\rm (c)}] $\SA(p^{-1}(H))$ is an epimorphic image of $\SA(H)$ for all $H \se A$ and $p\in \Tr(\calA)$.
  \item[{\rm (d)}] If $(\iota,\vp) : \calA \ra \calB$ is a g-morphism and $H \se B$, then $\RA(H\vpi) \preceq_g \RA(H)$. If $(\iota,\vp)$ is a g-epimorphism, then $\RA(H\vpi) \cong_g \RA(H)$.
\end{itemize}
\end{proposition}


\begin{proof} Claims (a)--(c) follow directly from the corresponding parts
of Proposition \ref{pr:SyntCongr and Operations} (in the same way
as for ranked algebras \cite{Stei92}).

To prove (d), assume first that $(\iota,\vp)$ is a
g-epimorphism. It follows from Proposition~\ref{pr:SyntCongr and
Operations}(d) that the maps $\psi:A/\theta_{H\vpi} \ra
B/\theta_H,\, a\theta_{H\vpi} \mapsto (a\vp)\theta_H,$ and
$\vk:\Si/\si_{H\vpi} \to \Om/\si_H, \, f\si_{H\vpi} \mapsto
\iota(f)\si_H,$
 are well-defined and injective. Clearly, they are also surjective, and for any $f \in \Si$, $m \ge 0$ and $a_1$, \dots, $a_m \in A$,
\begin{align*}
    (f\si_{H\vpi})_{\RA(H\vpi)} &(a_1\theta_{H\vpi}, \dots, a_m\theta_{H\vpi}) \psi
    = (f_{\calA}(a_1,\dots,a_m)\theta_{H\vpi})\psi \\
    &= (f_{\calA}(a_1,\dots,a_m)\vp)\theta_H
    = (\iota(f)_{\calB}(a_1\vp,\dots,a_m\vp))\theta_H \\
    &= (\iota(f)\si_H)_{\RA(H)}((a_1\vp)\theta_H,\dots,(a_m\vp)\theta_H) \\
    &= \vk(f\si_{H\vpi})_{\RA(H)}((a_1\theta_{H\vpi})\psi,\dots,(a_m\theta_{H\vpi})\psi),
\end{align*}
which shows that $(\vk,\psi): \RA(H\vpi) \ra \RA(H)$ is a
g-isomorphism.

Consider now a general g-morphism $(\iota,\vp) : \calA \ra
\calB$. Let $\calC = (C,\iota(\Si)/\sigma_H)$, where $C =
A\vp\vp_H$ and $\iota(\Si)/\si_H = \{\iota(f)\si_H \mid f \in
\Si\}$, be the image of $\calA$ in $\RA(H)$ under the g-morphism
$(\iota\iota_H,\vp\vp_H) : \calA \ra \RA(H)$. The mappings
$$
  \vk : \Si \ra \iota(\Si)/\si_H, \, f \mapsto \iota(f)\si_H, \; \mathrm{and} \; \psi : A \ra C, \, a \mapsto (a\vp)\theta_H,
$$
define a g-epimorphism $(\vk,\psi) : \calA \ra \calC$, and therefore $\RA(H\vpi\psi\psi^{-1}) \cong_g \RA (H\vpi\psi)$ by the previous part of the proof. Also, $\RA(H\vpi\psi) \preceq_g \RA(H)$ because $\RA(H\vpi\psi)$ is a g-image of the g-subalgebra $\calC$ of $\RA(H)$. To obtain $\RA(H\vpi) \preceq_g \RA(H)$ it therefore suffices to show that $H\vpi = H\vpi\psi\psi^{-1}$.
Of course, $H\vpi \se H\vpi\psi\psi^{-1}$, and on the other hand,
\begin{align*}
    H\vpi\psi\psi^{-1} &= (H\vpi)\vp\circ\vp_H\circ(\vp\circ\vp_H)^{-1} =(H\vpi)\vp\circ\theta_H\circ\vpi\\
    &\se (H\vpi)\theta_{H\vpi} = H\vpi,
\end{align*}
where we used the facts that $\psi = \vp\vp_H$,  $\vp_H\circ\vp_H^{-1} = \theta_H$, $\vp\circ\theta_H\circ\vpi \se \theta_{H\vpi}$ and $(H\vpi)\theta_{H\vpi} = H\vpi$.
\end{proof}

Simple modifications of the proofs of statements (d) of Propositions \ref{pr:SyntCongr and Operations} and \ref{pr:SyntAlg and Operations}
yield the following specializations of those statements.

\begin{corollary}\label{co:Morphisms and SA} Let $\SalgA$ and $\SalgB$ be unranked $\Si$-algebras. If $\vp : \calA \ra \calB$ is a morphism and $H\se B$, then $\vp\circ\theta_H\circ\vpi \se \theta_{H\vpi}$ and $\SA(H\vpi) \preceq \SA(H)$. If $\vp$ is an epimorphism, then  $\vp\circ\theta_H\circ\vpi = \theta_{H\vpi}$ and $\SA(H\vpi) \cong \SA(H)$.\ep
\end{corollary}

The \emph{syntactic congruence} $\theta_T$, the \emph{syntactic
algebra} $\SA(T)$, the \emph{syntactic morphism} $\vp_T : \SXta
\ra \SA(T)$, the \emph{reduced syntactic congruence}
$(\si_T,\theta_T)$ and the \emph{reduced syntactic algebra}
$\RA(T)$ as well as the \emph{syntactic g-morphism} $(\iota_T,\vp_T) :
\SXta \ra \RA(T)$ of a $\SX$-tree language $T$ are defined by
regarding $T$ as a subset of the term algebra $\SXta$. Since the
translations of $\SXta$ are given by  $\SX$-contexts, we have
$$s \, \theta_T \, t \: \LRa \: (\forall p\in \SXc)(p(s)\in T \lra p(t)\in T),$$
for all $s,t\in \SXt$. Let us note that the syntactic congruences
of unranked tree languages appear in \cite{BrMW01} under the name
``top congruence''.

\section{Recognizable unranked tree languages}\label{se:Recognizability}

Within our framework it is natural to define recognizability as
follows.

\begin{definition}\label{de:Recognizable Tree Languages}{\rm  An unranked
$\Si$-algebra $\SalgA$ \emph{recognizes} an unranked $\SX$-tree
language $T$ if there exist a morphism $\vp : \SXta \ra \calA$ and
a subset $F\se A$ such that $T = F\vp^{-1}$, and we call $T$
\emph{recognizable} if it is recognized by a regular
$\Si$-algebra. The set of all recognizable unranked $\SX$-tree
languages is denoted by $\RecSX$. }\ep
\end{definition}

Exactly as in the ranked case, the syntactic algebra of any given unranked
tree language $T$ is in a definite sense the minimal algebra
recognizing $T$.

\begin{proposition}\label{pr:SA(T) min rec of T} A $\Si$-algebra $\calA$ recognizes an unranked $\SX$-tree language $T$ if and only if $\SA(T) \preceq \calA$.
\end{proposition}

\begin{proof} It is clear, quite generally, that any tree language
recognized by a subalgebra or an epimorphic image of an algebra
$\calA$, is  recognized by $\calA$, too. Since $\SA(T)$ recognizes
$T$ (as $T = T\vp_T\vp_T^{-1}$), this means that if $\SA(T)
\preceq \calA$, then also $\calA$ recognizes $T$.
The converse holds by Corollary \ref{co:Morphisms and SA}: if there exist a morphism $\vp : \SXt \ra \calA$ and a subset $F$ of
$\calA$ such that $T = F\vpi$, then $\SA(T) = \SA(F\vpi) \preceq \SA(F) \preceq \calA$.
\end{proof}

It should be obvious that the above notion of recognizability is
the same as the one arrived at via deterministic or
nondeterministic automata on unranked trees as defined in
\cite{CrLT05,MaNi07,SeVi02}, for example.  Indeed, for any regular
algebra $\SalgA$, we can introduce finite (``horizontal'')
automata to compute the operations $f_\calA$ and, conversely, the
functions computed by the horizontal automata of a deterministic
automaton on unranked trees are the operations of a regular
algebra. However, since we are not directly concerned with
computational aspects, the algebraic definition is more convenient
here.

Next we give a Myhill-Nerode theorem that characterizes the
recognizability of unranked tree languages in terms of regular
congruences.

\begin{proposition}\label{pr:Myhill-Nerode Thm} For any unranked tree language $T\se \SXt$, the following statements are equivalent:
\begin{itemize}
  \item[{\rm (a)}] $T \in \RecSX$;
  \item[{\rm (b)}] $T$ is saturated by a regular congruence of $\SXta$;
  \item[{\rm (c)}] the syntactic congruence $\theta_T$ is regular.
\end{itemize}
\end{proposition}

\begin{proof} Let us first prove the equivalence of (a) and (b). If $T\in
\RecSX$, then there exist a regular algebra $\SalgA$, a morphism
$\vp : \SXta \ra \calA$ and subset $F\se A$ such that $T = F\vpi$.
It is clear that $T$ is saturated by $\theta := \ker \vp$.
Moreover, we may assume that $\vp$ is surjective. Then
$\SXta/\theta \cong \calA$ and hence $\theta \in \RCon(\SXta)$ by
Proposition \ref{pr:RegCon and RegAlg}. On the other hand, if $T$
is saturated by a regular congruence $\theta\in \RCon(\SXta)$,
then $T = T\theta_\natural\theta_\natural^{-1}$ means that $T$ is
recognized by the regular algebra $\SXta/\theta$.

If $T$ is saturated by a congruence $\theta\in\RCon(\SXta)$, then
$\theta \se \theta_T$ by Lemma \ref{le:SyntCongrSat}, and hence
$\theta_T$ is regular by Lemma \ref{le:FCon filter}. Therefore,
(b) implies (c), and the converse holds by Lemma
\ref{le:SyntCongrSat}.
\end{proof}

Let us note that in \cite{BrMW01} it was stated (as Lemma 8.2), in
different terms, that if $T \in \RecSX$, then $\theta_T$ is of
finite index, but the example meant to show that the converse does
not hold, appears incorrect. Nevertheless, their Theorem 1
essentially expresses the equivalence of (a) and (c) of our
Proposition \ref{pr:Myhill-Nerode Thm}; the condition concerning
``local views'' seems to be a somewhat intricate way to define the
regularity of a congruence.

The following corollary is an immediate consequence of Proposition
\ref{pr:Myhill-Nerode Thm}.

\begin{corollary}\label{co:Rec and SA} An unranked tree language $T$ is recognizable if and only if the syntactic algebra $\SA(T)$ is regular. \ep
\end{corollary}

We shall now note that the family of recognizable unranked tree
languages is closed under the operations that will define our
varieties of unranked tree languages.

\begin{proposition}\label{pr:ClosureProp of Rec} The following statements hold for all operator alphabets $\Si$ and $\Om$ and all leaf alphabets $X$ and $Y$.
\begin{itemize}
  \item[{\rm (a)}] $\es \in \RecSX$, and $\RecSX$ is closed under all Boolean operations.
  \item[{\rm (b)}]  If $T \in \RecSX$, then $p^{-1}(T) := \{t\in \SXt \mid p(t)\in T\}\in \RecSX$ for every context $p\in\SXc$.
  \item[{\rm (c)}] If  $(\iota,\vp) : \SXta \ra \OYta$ is a g-morphism, then $T\vpi \in \RecSX$ for every $T\in \RecOY$.
\end{itemize}
\end{proposition}

\begin{proof} Clearly, $\es$ and $\SXt$ are recognized by any $\Si$-algebra,
and the rest of the proposition follows from Corollary \ref{co:Rec
and SA} and Propositions \ref{pr:RegularAlgebrasVFRA} and
\ref{pr:SyntAlg and Operations}.
\end{proof}

The sets $p^{-1}(T)$ play an important role in the variety theory
and we shall need the following fact.

\begin{lemma}\label{le:Finite no of quotients} If $T\in Rec(\Si,X)$, then the set $\{p^{-1}(T) \mid p\in \SXc\}$ is finite.
\end{lemma}

\begin{proof} By Proposition \ref{pr:SyntCongr and Operations}(c) every set
$p^{-1}(T)$ is saturated by $\theta_T$. On the other hand, it
follows from Proposition~\ref{pr:Myhill-Nerode Thm} that
$\theta_T$ has just finitely many equivalence classes.  Hence, the
number of different sets $p^{-1}(T)$ must be finite, too.
\end{proof}

Let us say that a recognizable unranked $\SX$-tree language $T$ is
\emph{effectively given} if $T = F\vpi$, where $\vp : \SXta \ra
\calA$ is an effectively given morphism from $\SXta$ to an
effectively given regular algebra $\SalgA$ and $F\se A$ is also
effectively given.

\begin{proposition}\label{pr:SA(T) effectively computed} If $T \in \RecSX$ is effectively given, then $\SA(T)$ can be effectively constructed.
\end{proposition}

\begin{proof} Assume that $T = F\vpi$, where the morphism $\vp : \SXta \ra
\calA$, the regular algebra $\SalgA$ and the subset $F\se A$ are
effectively given. Obviously, we may assume that $\vp$ is an
epimorphism. Then $\SA(T) \cong \SA(F)$ by Proposition
\ref{pr:SyntAlg and Operations}, and $\SA(F)$ can be constructed
by Proposition \ref{pr:SA(H) effectively computed}.
\end{proof}

\section{Varieties of unranked tree languages}\label{se:VUTs and VRCs}

A \emph{family of (recognizable) unranked tree languages} is a
mapping $\calV$ that assigns to each pair $\Si$, $X$ a set
$\calV(\Si,X)$ of (recognizable) $\SX$-tree languages. We write
$\varV$ with the understanding that $\Si$ and $X$ range over all
operator alphabets and leaf alphabets, respectively. The inclusion
relation, unions and intersections of these families are defined
by the natural componentwise conditions. In particular, if $\varU$
and $\varV$ are two such families, then $\calU \se \calV$ means
that $\calU(\Si,X) \se \calV(\Si,X)$ for all $\Si$ and $X$, and
$\calU \cap \calV = \{\calU(\Si,X)\cap\calV(\Si,X)\}$.

\begin{definition}\label{de:VUT}{\rm A \emph{variety of unranked tree languages (VUT)} is a family of recognizable unranked tree languages  $\varV$ for which the following hold for all $\Si$, $\Om$, $X$ and $Y$.
\begin{itemize}
  \item[(V1)] $\es \neq \calV(\Si,X) \se \RecSX$.
  \item[(V2)] If $T\in \calV(\Si,X)$, then also $\SXt\setminus T$ belongs to $\calV(\Si,X)$.
  \item[(V3)] If $T,U\in \calV(\Si,X)$, then $T\cap U \in \calV(\Si,X)$.
  \item[(V4)] If $T\in \calV(\Si,X)$, then $p^{-1}(T) \in \calV(\Si,X)$ for every  $p\in \SXc$.
  \item[(V5)] If  $(\iota,\vp) : \SXta \ra \OYta$ is a g-morphism, then $T\vpi \in \calV(\Si,X)$ for every $T\in \calV(\Om,Y)$.
\end{itemize}
Let $\VUT$ denote the class of all VUTs.\ep}
\end{definition}

It is obvious that the intersection of any family of VUTs is a
VUT.
Hence, $(\VUT,\se)$ is a complete lattice. It is also clear that the
union of any directed family of VUTs is a VUT. The least VUT is
$Triv := \{\{\es,\SXt\}\}$ and the greatest one is the family $Rec
:= \{\RecSX\}$ of all recognizable unranked tree languages. The
following fact will be used in the proof of the variety theorem.

\begin{proposition}\label{pr:theta_T-classes in VUT} If $\varV$ is a VUT and $T\in \calV(\Si,X)$ for some $\Si$ and $X$, then every $\theta_T$-class is also in $\calV(\Si,X)$.
\end{proposition}

\begin{proof} It follows from the definition of $\theta_T$ that for any
$t\in \SXt$,
$$t\theta_T \: = \: \bigcap\{p^{-1}(T) \mid p\in \SXc, p(t) \in T\} \setminus
\bigcup\{p^{-1}(T) \mid p\in \SXc, p(t) \notin T\}.$$ By Lemma
\ref{le:Finite no of quotients}, this shows that $t\theta_T$ is in
$\calV(\Si,X)$.
\end{proof}

As in the ranked case, many VUTs have natural definitions based on
congruences of term algebras. Hence, before considering further
examples of VUTs, we introduce the systems of congruences that
yield varieties of unranked tree languages.  For any mapping $\vp
: A \ra B$ and any $\theta \in \Eq(B)$, let $\theta_\vp := \vp
\circ \theta \circ \vpi$. Then $\theta_\vp \in \Eq(A)$, and if
$B/\theta$ is finite, then so is $A/\theta_\vp$.

\begin{lemma}\label{GRC and g-morphism} If $(\om,\theta)\in \RGCon(\OYta)$, then $(\om_\iota,\theta_\vp) \in \RGCon(\SXta)$ for every g-morphism $(\iota,\vp) : \SXta \ra \OYta$.
\end{lemma}

\begin{proof} If $f,g\in \Si$, $s_1,\ldots,s_m,t_1,\ldots,t_m\in \SXt$
($m\geq 0$) are such that $f\, \om_\iota \,g$ and $s_i\,
\theta_\vp \, t_i$ for every $i \in [m]$, then $\iota(f) \, \om \,
\iota(g)$ and $s_i\vp\, \theta \, t_i\vp$ for every $i\in [m]$,
and therefore
$$
f(s_1,\ldots,s_m)\vp \: = \: \iota(f)(s_1\vp,\ldots,s_m\vp)
 \: \equiv_\theta \: \iota(g)(t_1\vp,\ldots,t_m\vp)
 \: = g(t_1,\ldots,t_m)\vp,
$$
which shows that $f_{\SXta}(s_1,\ldots,s_m) \, \theta_\vp \,
g_{\SXta}(t_1,\ldots,t_m)$ as required.
\end{proof}

By a \emph{family of regular g-congruences} we mean a mapping
$\calC$ that assigns to each pair $\Si$, $X$ a subset
$\calC(\Si,X)$ of $\RGCon(\SXta)$. Again, we write $\varC$ and
order these families by the componentwise inclusion relation.

\begin{definition}\label{de:VRCs}{\rm  A family of regular g-congruences $\varC$ is a \emph{variety of regular g-congruences (VRC)} if the following three conditions hold for all $\Si$, $\Om$, $X$ and $Y$.
\begin{itemize}
  \item[{\rm (C1)}] For every $\si \in \Eq(\Si)$,
  $\calC(\Si,X)_\si \, := \, \{\theta \in \RCon(\SXta) \mid (\si,\theta) \in \calC(\Si,X)\}$ is a filter of $\RCon(\SXta)$.
  \item[{\rm (C2)}] If $(\si,\theta)\in \calC(\Si,X)$ and $(\si',\theta)\in \RGCon(\SXta)$, then $(\si',\theta)\in \calC(\Si,X)$.
  \item[{\rm (C3)}] If  $(\iota,\vp) : \SXta \ra \OYta$ is any g-morphism, then $(\om_\iota,\theta_\vp)\in \calC(\Si,X)$ for every $(\om,\theta)\in \calC(\Om,Y)$.\ep
\end{itemize}
}
\end{definition}

We shall now show that any variety of regular congruences yields a
variety of unranked tree languages. For any family $\varC$ of
regular g-congruences, let $\calC^t$ be the family of recognizable
unranked tree languages such that for all $\Si$ and $X$,
$$\calC^t(\Si,X) \: := \: \{T\se \SXt \mid (\Delta_\Si,\theta_T) \in \calC(\Si,X)\}.$$

\begin{proposition}\label{pr:VRC to VUT} If $\varC$ is a VRC, then $\calC^t = \{\calC^t(\Si,X)\}$ is a VUT.
\end{proposition}

\begin{proof} Most of the proposition follows directly from the
definitions involved and Proposition \ref{pr:SyntCongr and
Operations}. Let us verify conditions (V1) and (V5) of Definition
\ref{de:VUT}.

Firstly, for any $\Si$ and $X$, $\calC^t(\Si,X) \neq \es$ because
$(\Delta_\Si,\nabla_{\SXt})$ certainly is in $\calC(\Si,X)$ and
$\theta_\es = \nabla_{\SXt}$. If $T\in \calC^t(\Si,X)$, then
$(\Delta_\Si,\theta_T)\in \calC(\Si,X)$ and hence $\theta_T \in
\RCon(\SXta)$, and by Proposition \ref{pr:Myhill-Nerode Thm} this
means that $T$ is recognizable. Hence, $\calC^t$ satisfies (V1).

If $(\iota,\vp) : \SXta \ra \OYta$ is any g-morphism and  $T\in
\calC^t(\Om,Y)$, then we have $(\Delta_\Om,\theta_T)\in \calC(\Om,Y)$
which, by condition (C3), implies
$$
(\iota\circ\Delta_\Om\circ\iota^{-1},\vp\circ\theta_T\circ\vpi)\in
\calC(\Si,X).
$$
Furthermore,
$(\Delta_\Si,\vp\circ\theta_T\circ\vpi)\in\RGCon(\SXta)$ because
$\theta_T \in \RCon(\OYta)$ implies that
$\vp\circ\theta_T\circ\vpi\in\RCon(\SXta)$. Hence,
$(\Delta_\Si,\vp\circ\theta_T\circ\vpi)\in\calC(\Si,X)$ by
condition (C2). On the other hand, $\vp\circ\theta_T\circ\vpi \se
\theta_{T\vpi}$ by Proposition \ref{pr:SyntCongr and
Operations}(d), and therefore
$(\Delta_\Si,\theta_{T\vpi})\in\calC(\Si,X)$ by (C1). This means
that $T\vpi \in \calC^t(\Si,X)$ and therefore $\calC^t$ satisfies
(V5).
\end{proof}

In many of the following examples, we define a whole family
$\calF$ of VUTs indexed by some parameter(s). If $\calF$ forms an
ascending chain or, more generally, is directed, then the union
$\bigcup \calF$ is always also a VUT. In the ranked case the basic
varieties of tree languages forming such a family $\calF$ are
usually defined by so-called \emph{principal varieties of
congruences} \cite{Stei92,Stei05} that consist of principal
filters of the term algebras. Here the same purpose will be served
by the following more general notion.

\begin{definition}\label{de:Consistent system}{\rm For each pair $\Si$, $X$, let $\theta(\Si,X)$ be  a congruence of $\SXta$. We call $\Theta = \{\theta(\Si,X)\}_{\Si,X}$ a \emph{consistent system of congruences} if $\theta(\Si,X) \,\se\, \vp\circ\theta(\Om,Y)\circ\vpi$ for all alphabets $\Si$, $\Om$, $X$ and $Y$, and every g-morphism $(\iota,\vp) : \SXta \ra \OYta$.

For any such system of congruences $\Theta =
\{\theta(\Si,X)\}_{\Si,X}$, and  all $\Si$ and $X$,  let
$$\calC_\Theta(\Si,X) \, := \, \{(\si,\theta) \in \RGCon(\SXta)
\mid \theta(\Si,X) \se \theta \},$$ and let $\calC_\Theta :=
\{\calC_\Theta(\Si,X)\}$ be the thus defined family of regular
congruences. }\ep
\end{definition}

\begin{lemma}\label{le:Congruences to VCRs}  For any consistent system of congruences $\Theta$, $\calC_\Theta$ is a VRC.
\end{lemma}

\begin{proof} (C1) If $\si \in \Eq(\Si)$, then $(\si,\nabla_{\SXt})\in
\calC_\Theta(\Si,X)$, and hence $\calC_\Theta(\Si,X)_\si \neq
\es$. If $\theta \se \rho$ and $\theta \in
\calC_\Theta(\Si,X)_\si$, then $\theta(\Si,X) \se \theta \se
\rho$. On the other hand, $(\si,\theta)\in \RGCon(\SXta)$ implies
$(\si,\rho)\in \RGCon(\SXta)$ by Lemma \ref{le:FCon filter}. Hence
$\rho \in \calC_\Theta(\Si,X)_\si$. If $\theta,\rho \in
\calC_\Theta(\Si,X)_\si$, then $\theta(\Si,X) \se \theta,\rho$ and
$(\si,\theta),(\si,\rho) \in \RGCon(\SXta)$, and therefore
$\theta(\Si,X) \se \theta\cap\rho$ and -- again by Lemma
\ref{le:FCon filter}, $(\si,\theta\cap\rho) =
(\si,\theta)\wedge(\si,\rho) \in \RGCon(\SXta)$. This means that
$\theta\cap\rho \in \calC_\Theta(\Si,X)_\si$, and thus we have
shown that $\calC_\Theta(\Si,X)_\si$ is a filter in
$\RCon(\SXta)$.

(C2) If $(\si,\theta)\in \calC_\Theta(\Si,X)$ and
$(\si',\theta)\in \RGCon(\SXta)$, then also $(\si',\theta)\in
\calC_\Theta(\Si,X)$ because $\theta(\Si,X) \se \theta$ by the
first assumption.

(C3) If  $(\iota,\vp) : \SXta \ra \OYta$ is a g-morphism and
$(\om,\theta)\in \calC_\Theta(\Om,Y)$, then
$(\om_\iota,\theta_\vp)\in \RGCon(\SXta)$ by Lemma \ref{GRC and
g-morphism}, and $\theta(\Si,X) \se \theta(\Om,Y)_\vp \se
\theta_\vp$ by our assumption about the
$\theta(\Si,X)$-congruences and the fact that $\theta(\Om,Y) \se
\theta$. Hence $(\om_\iota,\theta_\vp)\in \calC_\Theta(\Si,X)$.
\end{proof}

Let us call a VRC \emph{quasi-principal} if it is defined this way
by a consistent system of congruences $\Theta$. The corresponding
VUT $\calC_\Theta^t$ is written as $\calV_\Theta =
\{\calV_\Theta(\Si,X)\}$ and also it is said to be
\emph{quasi-principal}. The following description of
$\calV_\Theta$ is a direct consequence of its definition.

\begin{lemma}\label{Quasi-principal VUT} Let $\Theta = \{\theta(\Si,X)\}_{\Si,X}$ be  any consistent system of congruences.
Then $\calV_\Theta(\Si,X) \, = \, \{T\in \RecSX \mid \theta(\Si,X)
\se \theta_T\}$ for all $\Si$ and $X$. \ep
\end{lemma}


\section{The variety theorem}\label{se:Variety theorem}

We shall now prove that the following maps $\bfK \mapsto \bfK^t$
and $\calV \mapsto \calV^a$ form a pair of mutually inverse
isomorphisms between the lattices $(\VRA,\se)$ and $(\VUT,\se)$.

\begin{definition}\label{de:VRA to VUT to VRA}{\rm For any VRA $\bfK$, let $\bfK^t = \{\bfK^t(\Si,X)\}$ be the family of recognizable unranked tree languages in which, for all $\Si$ and $X$,
$$\bfK^t(\Si,X) \, := \, \{T\se \SXt \mid  \SA(T) \in \bfK\}.$$
For any VUT $\varV$, let $\calV^a$ be the VRA generated by the
algebras $\SA(T)$, where $T\in \calV(\Si,X)$ for some $\Si$ and
$X$. }\ep
\end{definition}

Note that $\calV^a$ is a well-defined VRA for every VUT $\calV$
because all the algebras $\SA(T)$ with $T\in \calV(\Si,X)$ are
regular. Note also that the maps $\bfK \mapsto \bfK^t$ and $\calV
\mapsto \calV^a$ were defined in terms of syntactic algebras, but
it follows from Lemma \ref{le:SAs, RAs and VRAs} that definitions
that use reduced syntactic algebras (similarly as in
\cite{Stei98}) would give the same maps.

\begin{lemma}\label{le:bfKt is a VUT} For any VRA $\bfK$, $\bfK^t$ is a VUT.
\end{lemma}

\begin{proof} It follows from Corollary \ref{co:Rec and SA}  that
$\bfK^t(\Si,X) \se Rec(\Si,X)$ for all $\Si$ and $X$. Moreover,
$\bfK^t(\Si,X) \neq \es$ because $\bfK$ contains at least the
trivial $\Si$-algebras. Hence, $\bfK^t$ satisfies condition (V1)
of Definition \ref{de:VUT}. Conditions (V2)--(V4) follow
immediately from Proposition \ref{pr:SyntAlg and Operations} and
the fact that $\bfK$ is a VRA. As to (V5), we may argue as
follows. If $T\in \bfK^t(\Om,Y)$, then $\SA(T)\in \bfK$. By Lemma
\ref{le:SAs, RAs and VRAs} this implies $\RA(T)\in \bfK$ which by
Proposition \ref{pr:SyntAlg and Operations}(d) implies that
$\RA(T\vpi)\in \bfK$. Using again Lemma \ref{le:SAs, RAs and
VRAs}, we get $\SA(T\vpi)\in \bfK$ from which $T\vpi \in
\bfK^t(\Si,X)$ follows.
\end{proof}

It is clear that the maps $\bfK \mapsto \bfK^t$ and $\calV \mapsto
\calV^a$ are isotone. To prove that they define isomorphisms
between the lattices  $(\VRA,\se)$ and $(\VUT,\se)$ it therefore
suffices to show that they are inverses of each other.

\begin{lemma}\label{le:VRA to VUT to VRA}  $\bfK^{ta} = \bfK$ for every VRA $\bfK$.
\end{lemma}

\begin{proof} The VRA $\bfK^{ta}$ is generated by the algebras $\SA(T)$,
where $T\in \bfK^t(\Si,X)$ for some $\Si$ and $X$, but these
algebras are, by the definition of $\bfK^t$, also in $\bfK$.
Hence, $\bfK^{ta} \se \bfK$.

On the other hand, by Proposition \ref{pr:SDS-irr is syntactic},
$\bfK$ is generated by regular syntactic algebras. Let $\SalgA$ be
any such generating algebra. If $X$ is a sufficiently large leaf
alphabet, there is an epimorphism $\vp : \SXta \ra \calA$.
Furthermore, $\calA$ has a disjunctive subset $D$ by Proposition
\ref{pr:SyntAlg and DisjSubs SdiAlg}. The $\SX$-tree language $T
:= D\vpi$ is recognizable, and  $\SA(T) \cong \SA(D)$ by Corollary
\ref{co:Morphisms and SA}. On the other hand,  $\calA \cong
\SA(D)$ because $D$ disjunctive, and therefore also $\SA(T)\in
\bfK$, which shows that $T \in \bfK^t(\Si,X)$. As this means that
$\SA(T) \in \bfK^{ta}$, we also get $\calA \in \bfK^{ta}$ and we
can conclude that $\bfK\se\bfK^{ta}$.
\end{proof}

\begin{lemma}\label{le:VUT to VRA to VUT}  $\calV^{at} = \calV$ for every VUT $\calV$.
\end{lemma}

\begin{proof} If $T\in \calV(\Si,X)$, then $\SA(T)\in \calV^a$ implies $T\in
\calV^{at}(\Si,X)$, and hence $\calV \se \calV^{at}$.

If $T\in \calV^{at}(\Si,X)$, then $\SA(T)\in \calV^a$ and by
Proposition \ref{pr:V_g = HSP_g} this means that
$$
    \SA(T) \, \preceq \, \vk(\SA(U_1),\ldots,\SA(U_n)),
$$
where $n\geq 0$, $U_1\in \calV(\Si_1,X_1)$, \dots, $U_n\in
\calV(\Si_n,X_n)$ for some alphabets $\Si_1,\ldots,\Si_n$ and
$X_1,\ldots,X_n$, and $\vk$ is a mapping from $\Si$ to
$\Si_1\times\ldots\times\Si_n$.

For each $i\in [n]$, denote $\calT_{\Si_i}(X_i)$ by $\calT_i$, and
let $\SA(U_i) = (A_i,\Si_i)$. Furthermore, let $\vp_i : \calT_i
\ra \SA(U_i), t \mapsto t\theta_{U_i},$ be the syntactic morphism
of $U_i$. By Proposition \ref{pr:SA(T) min rec of T}, there exist
a morphism
$$\vp : \SXta \ra \vk(\SA(U_1),\ldots,\SA(U_n))$$
and a subset $F\se A_1\times \ldots \times A_n$ such that $T =
F\vpi$. For each $i\in [n]$, define $\lambda_i : \Si \ra \Si_i$ so
that for any $f\in \Si$, if $\vk(f) = (f_1,\ldots,f_n)$, then
$\lambda_i(f) = f_i$. The syntactic morphisms $\vp_i$ yield an
epimorphism
$$\eta : \vk(\calT_1,\ldots,\calT_n) \ra \vk(\SA(U_1),\ldots,\SA(U_n)), (t_1,\ldots,t_n) \mapsto (t_1\vp_1,\ldots,t_n\vp_n),$$
and for each $i\in [n]$, we get the g-morphisms
$$(\lambda_i,\tau_i) : \vk(\calT_1,\ldots,\calT_n) \ra \calT_i \;\, \mathrm{and} \;\, (\lambda_i,\pi_i) : \vk(\SA(U_1),\ldots,\SA(U_n)) \ra \SA(U_i),$$
where $\tau_i : (t_1,\ldots,t_n) \mapsto t_i$ and $\pi_i :
(a_1,\ldots,a_n) \mapsto a_i$ are the respective $i^{th}$
projections.

Clearly, $\tau_i\vp_i = \eta\pi_i$ for every $i\in [n]$. Since
$\eta$ is surjective, we may define a mapping $\psi_0 : X \ra
T(\Si_1,X_1)\times \ldots \times T(\Si_n,X_n)$ such that
$x\psi_0\eta = x\vp$ for every $x\in X$. If $\psi : \SXta \ra
\vk(\calT_1,\ldots,\calT_n)$ is the homomorphic extension of
$\psi_0$, then $\psi\eta = \vp$.

Now, $T$ is the union of finitely many sets $\bfa\vpi$
with $\bfa = (a_1,\ldots,a_n)\in F$. Since $\vp\pi_i =
\psi\tau_i\vp_i$ for each $i\in [n]$, we have
$$\bfa\vpi \, = \, \bigcap\{a_i(\vp\pi_i)^{-1} \mid i\in [n]\} \, = \ \bigcap\{(a_i\vp_i^{-1})(\psi\tau_i)^{-1} \mid i \in [n]\},$$
where each $a_i\vp_i^{-1}$ is a $\theta_{U_i}$-class, and
therefore belongs to $\calV(\Si_i,X_i)$ by Proposition
\ref{pr:theta_T-classes in VUT}. By the definition of VUTs, this
means that $(a_i\vp_i^{-1})(\psi\tau_i)^{-1} \in \calV(\Si,X)$ for
every $i\in [n]$, and hence also $T\in \calV(\Si,X)$. This
concludes the proof of $\calV^{at}\se \calV$.
\end{proof}

The above results can be summed up as the following  variety
theorem.

\begin{theorem}\label{th:Variety Theorem} The mappings
$\VRA \ra \VUT$, $\bfK \mapsto \bfK^t$, and $\VUT \ra \VRA$,
$\calV \mapsto \calV^a$, are mutually inverse isomorphisms between
the lattices $(\VRA,\se)$ and $(\VUT,\se)$. \ep
\end{theorem}

\section{Examples of varieties of unranked tree languages}\label{se:Examples of VUTs}

We shall now introduce several varieties of unranked tree languages.
Most of them correspond to some general variety of ranked tree languages
considered in  \cite{Stei98}. The following simple observations
concerning g-morphisms of term algebras are helpful in many of the
examples.

\begin{lemma}\label{le:Trees and g-morphisms} Let
$(\iota,\vp) : \SXta \ra \OYta$ be a g-morphism. \begin{itemize}
\item[{\rm (a)}]  $\hg(f\vp) = 0$ and $\root(f\vp) = f\vp =
\iota(f)$ for every $f\in \Si$. \item[{\rm (b)}] If $t =
f(t_1,\ldots,t_m)$ ($m>0$), then $t\vp =
\iota(f)(t_1\vp,\ldots,t_m\vp)$ and $\root(t\vp) = \iota(f)$.
\item[{\rm (c)}] $\hg(t\vp) \geq \hg(t)$ for every $t\in \SXt$.
\item[{\rm (d)}] $t\vpi$ is finite for every $t\in \OYt$. \ep
\end{itemize} \end{lemma}

All statements of Lemma \ref{le:Trees and g-morphisms} are
obvious. Note, however, that (d) does not follow directly from
(c), as in the ranked case.

\subsection{Nilpotent unranked tree languages}

For any $\Si$ and
$X$, let $Nil(\Si,X)$ consist of all finite $\SX$-tree languages
and their complements in $\SXt$, and let  $Nil := \{Nil(\Si,X)\}$.
In view of Proposition \ref{pr:ClosureProp of Rec}(a), to prove
that $Nil \se Rec$, it suffices to note that, for all $\Si$ and
$X$, each singleton set $\{t\}\se \SXt$ is recognizable.  It is
clear that each set $Nil(\Si,X)$ is closed under all Boolean
operations, and condition (V4) and (V5) follow from the facts that
the pre-images $p^{-1}(T)$ and $T\vpi$ are finite for any finite
$T$; for $p^{-1}(T)$ this is obvious and for $T\vpi$ it follows
from Lemma \ref{le:Trees and g-morphisms}(d).

Similarly as in the unranked case \cite{Stei92,Stei05}, it
is easy to find the VRA corresponding to $Nil$. However, nilpotent
unranked algebras cannot be defined just in terms of the height of
trees as there are infinitely many trees of any height $\geq 1$.
Let us define the  \emph{size} $\size(t)$ of tree $t\in \SXt$ as
the number of nodes of $t$, i.e.,
\begin{itemize}
  \item[(1)] $\size(t) = 1$ for $t\in \Sigma \cup X$, and
  \item[(2)] $\size(t) = \size(t_1)+\ldots+\size(t_m)+1$ for $t = f(t_1,\ldots,t_m)$.
\end{itemize}

We call an unranked algebra $\SalgA$ \emph{nilpotent} if there
exist an element $a_0\in A$ and a $k\geq 1$ such that for any $X$
and $t\in \SXt$, if $\size(t) \geq k$, then $t^\calA(\alpha) =
a_0$ for every $\alpha : X \ra A$. The element $a_0$ is then
called the \emph{absorbing state} of $\calA$ and the least $k$ for
which the above condition holds is its \emph{degree (of
nilpotency)}. For each $k\geq 1$, let $\Nil_k$ denote the class of
regular nilpotent algebras of degree $\leq k$, and let $\Nil$ be
the class of all regular nilpotent algebras. Obviously, $\Nil_1
\subset \Nil_2 \subset \Nil_3 \subset \ldots$ and $\Nil =
\bigcup_{k\geq 1} \Nil_k$. Let us now show that each $\Nil_k$, and
hence also $\Nil$, is a VRA. For this consider any $k\geq 1$ and
any regular algebras $\SalgA$ and $\OalgB$, and assume that $\calA
\in \Nil_k$ with absorbing state $a_0$.

It is obvious that if $\calB$ is a g-subalgebra of $\calA$, then
also $\calB \in \Nil_k$ with $a_0$ as the absorbing state.
Next, let $(\iota,\vp) : \calA \ra \calB$ be a g-epimorphism. For
any $X$ and any $\beta : X \ra B$, there is a mapping $\alpha : X
\ra A$ such that $\alpha\vp = \beta$. Recall the mapping $\iota_X
: \SXt \ra \OXt$ defined in Section \ref{se:Unranked algebras}.
It is easy to verify by induction on $s$ that
$s^\calA(\alpha)\vp = \iota_X(s)^\calB(\beta)$ for every $s\in \SXt$.
Now, let $t$ be any $\OX$-tree of size $\geq k$. As $\iota$ is
surjective, there exists an $s\in \SXt$ such that $\iota_X(s) =
t$, and hence $t^\calB(\beta) = \iota_X(s)^\calB(\beta) =
s^\calA(\alpha)\vp = a_0\vp$. This shows that $\calB$ is in
$\Nil_k$ with $a_0\vp$ as its absorbing state.

Assume now that also $\calB\in\Nil_k$ and let $b_0$ be the
absorbing state.  Consider any g-product $\vk(\calA,\calB) =
(A\times B,\Ga)$, and any $X$ and  $\gamma : X \ra A\times B$. Let
us define $\alpha : X \ra A$, $\beta : X \ra B$, $\iota : \Ga \ra
\Si$ and $\lambda : \Ga \ra \Om$ by $\alpha := \gamma\pi_1$,
$\beta := \gamma\pi_2$, $\iota := \vk\pi_1$ and $\lambda :=
\vk\pi_2$, respectively. If $t\in \GXt$, then $\iota_X(t) \in
\SXt$, $\lambda_X(t) \in\OXt$ and $\size(\iota_X(t)) =
\size(\lambda_X(t)) = \size(t)$, and it can be verified by
induction on $t$ that $t^{\vk(\calA,\calB)}(\gamma) =
(\iota_X(t)^\calA(\alpha),\lambda_X(t)^\calB(\beta))$. This means
that if $\size(t) \geq k$, then $t^{\vk(\calA,\calB)}(\gamma) =
(a_0,b_0)$ which shows that $\vk(\calA,\calB)\in\Nil_k$.

It remains to be shown that $\Nil^t = Nil$. First, let  $T \se
\SXt$ be recognized by some $\SalgA$ in $\Nil_k$, i.e., $T =
F\vpi$ for some morphism $\vp : \SXta \ra \calA$ and $F \se A$. If
$a_0$ is the absorbing state of $\calA$ and $\alpha : X \ra A$ is
the restriction of $\vp$ to $X$, then $t^\calA(\alpha) = a_0$ for
every $t\in \SXt$ of size $\geq k$. This means that $T$ is finite
if $a_0\notin F$ and co-finite if $a_0\in F$. Hence, $\Nil^t \se
Nil$.

To prove the converse inclusion, consider any $\Si$, $X$ and a
finite $\SX$-tree language $T$. Let $k := \max\{\size(t) \mid t\in
T\}+1$ (for $T = \es$, let $k=1$). We construct a nilpotent
algebra $\SalgA$ recognizing $T$ as follows. Let $B := \{t\in \SXt
\mid \size(t) < k\}$ and $A := B \cup \{a_0\}$ (with $a_0\notin
B$), and for all $f\in \Si$, $m\geq 0$ and $b_1,\ldots,b_m\in A$
set
\begin{displaymath}
f_\calA(b_1,\ldots,b_m) = \left\{ \begin{array}{ll}
                f(b_1,\ldots,b_m) & \textrm{if $f(b_1,\ldots,b_m) \in B$;} \\
                a_0 & \textrm{otherwise.}
                \end{array} \right.
\end{displaymath}
It is clear that $t^\calA(\alpha) = a_0$ for every $\alpha : X \ra
A$ whenever $t\in \SXt$ and $\size(t) \geq k$, and also that
$\calA$ is regular. If $\vp : \SXta \ra \calA$ is the morphism
such that $x\vp = x$ for every $x\in X$, then $t\vp =
t^\calA(\alpha) = t$ if $\size(t) < k$ and $t\vp = a_0$ otherwise.
This means that $T = T\vpi$. For a co-finite $T$, we construct
such an $\calA$ for $S := \SXt \setminus T$ and obtain $T$ as
$(A\setminus S)\vpi$. Hence, $Nil \se \Nil^t$.

The above findings may be summed up as follows.

\begin{proposition}\label{pr:Nil} $Nil$ is the VUT corresponding to the VRA $\Nil$.  \ep
\end{proposition}

\subsection{Definite unranked tree languages}

The \emph{$k$-root} $\rt_k(t)$ of a $\SX$-tree $t$ is defined as follows:
\begin{itemize}
  \item[(0)] $\rt_0(t) = \ve$, where $\ve$ represents the empty root segment, for all $t\in \SXt$;
  \item[(1)] $\rt_1(t) = \root(t)$ for every $t\in \SXt$;
  \item[(2)] for $k\geq 2$, $\rt_k(t) = t$ if $\hg(t) < k$, and $\rt_k(t) = f(\rt_{k-1}(t_1),\ldots,\rt_{k-1}(t_m))$ if $\hg(t) \geq k$ and $t = f(t_1,\ldots,t_m)$.
\end{itemize} We call a recognizable unranked $\SX$-tree language $T$
\emph{$k$-definite} if for all $s,t\in \SXt$, if  $\rt_k(s) =
\rt_k(t)$ and  $s\in T$, then $t\in T$, and it is \emph{definite}
if it is $k$-definite for some $k\geq 0$. Let  $Def_k =
\{Def_k(\Si,X)\}$ and  $Def = \{Def(\Si,X)\}$  be the families of
$k$-definite  ($k\geq 0$) and all  definite tree languages,
respectively. Clearly $Def_0 \subset Def_1 \subset Def_2 \subset
\ldots$ and $Def = \bigcup_{k\geq 0}Def_k$.
We could naturally verify directly that the families $Def_k$ satisfy conditions (V1)--(V5), but let us show how
they are obtained from consistent systems of congruences. For any
$k\geq 0$, $\Si$ and $X$, define the relation $\delta_k(\Si,X)$ in
$\SXt$ by
$$
    s \, \delta_k(\Si,X) \, t \: :\LRa \: \rt_k(s) = \rt_k(t) \qquad (s,t\in \SXt).
$$
Note that for every  $k\geq 2$, there are infinitely many
$\delta_k(\Si,X)$-classes. Let $\Delta(k) :=
\{\delta_k(\Si,X))\}_{\Si,X}$. The following technical lemma is needed for showing that the families $Def_k$ are VUTs.

\begin{lemma}\label{le:g-morphisms and roots} Let $(\iota,\vp) : \SXta \ra \OYta$ be a g-morphism. Then $\rt_k(t\vp) = \rt_k(\rt_k(t)\vp)$ for all $t\in \SXt$ and $k\geq 1$.
\end{lemma}

\begin{proof} We proceed by induction on $k\geq 1$. The case $k=1$ is obvious: $\rt_1(t\vp) = \root(t\vp) = \root(\root(t\vp)) = \rt_1(\rt_1(t)\vp)$.

Assume now that $k\geq 2$ and that the lemma holds for all smaller values of $k$. If $\hg(t) < k$, then $\rt_k(t) = t$ and hence $\rt_k(\rt_k(t)\vp) = \rt_k(t\vp)$. Assume that $\hg(t) \geq k$ and let $t = f(t_1,\ldots,t_m)$. Then $t\vp = \iota(f)(t_1\vp,\ldots,t_m\vp)$ and therefore
\begin{align*}
    \rt_k(t\vp) \: &= \: \iota(f)(\rt_{k-1}(t_1\vp),\ldots,\rt_{k-1}(t_m\vp))\\
    &= \: \iota(f)(\rt_{k-1}(\rt_{k-1}(t_1)\vp),\ldots,\rt_{k-1}(\rt_{k-1}(t_m)\vp))\\
&= \: \rt_k(\iota(f)(\rt_{k-1}(t_1)\vp,\ldots,\rt_{k-1}(t_m)\vp))\\
&= \: \rt_k(f(\rt_{k-1}(t_1),\ldots,\rt_{k-1}(t_m))\vp)\\
&= \: \rt_k(\rt_k(t)\vp),
\end{align*}
where we also used the inductive assumption.
\end{proof}

\begin{proposition}\label{pr:Def_k VUT} For every $k\geq 0$,  $\Delta(k)$ is a consistent system of congruences, and $Def_k$ is the quasi-principal VUT defined by it, and hence  also $Def$ is a VUT.
\end{proposition}

\begin{proof} Obviously it suffices to show that the following statements
(a)--(c) hold for every $k\geq 0$ and for all alphabets $\Si$,
$\Om$, $X$, $Y$.
\begin{itemize}
    \item[{\rm (a)}] $\delta_k(\Si,X)$ is a congruence of $\SXta$.
    \item[{\rm (b)}] $\delta_k(\Si,X) \,\se\, \vp\circ\delta_k(\Om,Y)\circ\vpi$  for any g-morphism $(\iota,\vp) : \SXta \ra \OYta$.
    \item[{\rm (c)}] A recognizable $\SX$-tree language $T$ is $k$-definite iff $\delta_k(\Si,X) \se \theta_T$.
\end{itemize}
 Since $\delta_0(\Si,X) = \nabla_{\SXt}$, statement (a) trivially holds for $k=0$. For $k>0$, it follows from the obvious fact that $\rt_k(s) = \rt_k(t)$ implies $\rt_{k-1}(s) = \rt_{k-1}(t)$.

To prove (b) it suffices to show that if $s,t\in \SXt$ and $\rt_k(s) = \rt_k(t)$, then $\rt_k(s\vp) = \rt_k(t\vp)$, and this follows from Lemma \ref{le:g-morphisms and roots}.

Let $T\in \RecSX$. Since $T$ being $k$-definite means precisely
that $T$ is saturated by $\delta_k(\Si,X)$, statement (c) follows
from Lemma \ref{le:SyntCongrSat}.

As the union of the chain $Def_0 \subset Def_1 \subset Def_2
\subset \ldots$, also $Def$ is a VUT.
\end{proof}

\subsection{Reverse definite unranked tree languages}

A $\SX$-tree $s$ is a \emph{subtree} of a $\SX$-tree $t$ if $t =
p(s)$ for some  context  $p \in \SXc$. For any $t\in \SXt$, let
$\st(t)$ denote the set of subtrees of $t$, and for each $k \geq
0$, let $\st_k(t) = \{ s \in \st(t) \mid  \hg(s) < k \} $. Note
that $\st_0(t) = \es$ for every $t$.

We call a recognizable unranked $\SX$-tree language $T$
\emph{reverse $k$-definite} if for all $s,t\in \SXt$, if $\st_k(s)
= \st_k(t)$ and  $s\in T$, then $t\in T$, and it is \emph{reverse
definite} if it is reverse $k$-definite for some $k\geq 0$. Let
$RDef_k = \{RDef_k(\Si,X)\}$ and  $RDef = \{RDef(\Si,X)\}$  be the
families of  $k$-reverse definite ($k\geq 0$) and all reverse
definite tree languages. Clearly $RDef_0 \subset RDef_1 \subset
RDef_2 \subset \ldots$ and $RDef = \bigcup_{k\geq 0}RDef_k$.

For each $k\geq 0$, a consistent system of congruences
$\mathrm{P}(k) = \{\rho_k(\Si,X)\}_{\Si,X}$ defining $RDef_k$  is
obtained when we set for any $\Si$ and $X$,
$$
    s \, \rho_k(\Si,X) \, t \: :\LRa \: \st_k(s) = \st_k(t) \qquad
    (s,t\in \SXt).
$$

To prove the consistency of the systems $\mathrm{P}(k)$, we need
the following fact.

\begin{lemma}\label{le:g-morph and subtrees} Let $(\iota,\vp) : \SXta \ra \OYta$ be a g-morphism. Then  $\st_k(t\vp) = \bigcup\{\st_k(s\vp) \mid s\in \st_k(t)\}$  for all $t\in \SXt$ and $k\geq 0$.
\end{lemma}

\begin{proof} The equality is obvious when $k=0$, so we assume that $k>0$
and proceed by induction on $t\in \SXt$.

If $t\in \Si\cup X$, then the equality certainly holds because
$\st_k(t) = \{t\}$.

Let $t = f(t_1,\ldots,t_m$), where $m>0$, and assume that the
claim holds for all trees of height $< \hg(t)$. If $\hg(t) < k$,
then again $t \in \st_k(t)$. If $\hg(t) \geq k$, then $\st_k(t)
= \st_k(t_1)\cup \ldots \cup \st_k(t_m)$. Since we also have $t\vp
= \iota(f)(t_1\vp,\ldots,t_m\vp)$, we get
\begin{align*}
\st_k(t\vp) \, &= \, \st_k(t_1\vp)\cup \ldots\cup \st_k(t_m\vp)\\
     & = \, \bigcup\{\st_k(s\vp) \mid s\in \st_k(t_1)\}\cup \ldots \cup \bigcup\{\st_k(s\vp) \mid s\in \st_k(t_m)\}\\
     & = \, \bigcup\{\st_k(s\vp) \mid s\in \st_k(t)\}
\end{align*}
by applying the inductive assumption to the trees
$t_1,\ldots,t_m$.
\end{proof}

\begin{proposition}\label{pr:RDef_k VUT} For every $k\geq 0$, $\mathrm{P}(k)$ is a consistent system of congruences, and $RDef_k$ is the quasi-principal VUT defined by it, and hence also $RDef$ is a VUT.
\end{proposition}

\begin{proof} We should show that the following hold for all  $k\geq 0$,
$\Si$, $\Om$, $X$, and $Y$.
\begin{itemize}
    \item[{\rm (a)}] $\rho_k(\Si,X)$ is a congruence of $\SXta$.
    \item[{\rm (b)}] $\rho_k(\Si,X) \,\se\, \vp\circ\rho_k(\Om,Y)\circ\vpi$  for any g-morphism $(\iota,\vp) : \SXta \ra \OYta$.
    \item[{\rm (c)}] A recognizable $\SX$-tree language $T$ is reverse $k$-definite iff $\rho_k(\Si,X) \se \theta_T$.
\end{itemize}
Statement (a) holds for $k=0$ as $\rho_0(\Si,X) = \nabla_{\SXt}$.
Let $k>0$ and consider any $f\in \Si$, $m>0$ and any
$s_1,\ldots,s_m,t_1,\ldots,t_m\in \SXt$ such that
$(s_1,t_1)$, \dots, $(s_m,t_m)\in \rho_k(\Si,X)$. We distinguish two
cases. If $\hg(f(s_1,\ldots,s_m)) < k$, then
$\hg(s_1),\ldots,\hg(s_m) < k$ and we must have
$s_1=t_1,\ldots,s_m= t_m$, from which
$f(s_1,\ldots,s_m)\,\rho_k(\Si,X)\,f(t_1,\ldots,t_m)$ trivially
follows. On the other hand, if $\hg(f(s_1,\ldots,s_m)) \ge k$,
then $\st_k(s_1) = \st_k(t_1),\ldots,\st_k(s_m) = \st_k(t_m)$ implies
$$
    \st_k(f(s_1,\ldots,s_m)) = \st_k(s_1)\cup\ldots\cup\st_k(s_m) = \st_k(f(t_1,\ldots,t_m)),
$$
and hence again $f(s_1,\ldots,s_m)\,\rho_k(\Si,X)\,f(t_1,\ldots,t_m)$.

To prove (b) it suffices to show that if $\st_k(s) = \st_k(t)$ for
some $s,t\in \SXt$, then $\st_k(s\vp) = \st_k(t\vp)$, but this
holds by Lemma \ref{le:g-morph and subtrees}.

Let $T\in \RecSX$. Since $T$ being reverse $k$-definite means
precisely that $T$ is saturated by $\rho_k(\Si,X)$, statement (c)
follows from Lemma \ref{le:SyntCongrSat}.

As the union of the chain $RDef_0 \subset RDef_1 \subset RDef_2
\subset \ldots$, also $RDef$ is a VUT.
\end{proof}

\subsection{Generalized definite tree languages}

For any $h,k\geq 0$, we call an unranked $\SX$-tree language $T$
\emph{$h,k$-definite} if for all $s,t\in \SXt$, if $\st_h(s) =
\st_h(t)$ and $\rt_k(s) = \rt_k(t)$, then  $s\in T$ iff $t\in T$,
and it is \emph{generalized definite} if it is $h,k$-definite for
some $h,k\geq 0$. Let  $GDef_{h,k} = \{GDef_{h,k}(\Si,X)\}$ and
$GDef = \{GDef(\Si,X)\}$  be the families of all recognizable
$h,k$-definite ($h,k\geq 0$) and all recognizable general definite
tree languages. Clearly $GDef_{h,k} \se GDef_{h',k'}$ whenever
$h\leq h'$ and $k\leq k'$, and $GDef = \bigcup_{h,k\geq
0}GDef_{h,k}$.

For any $h,k\geq 0$, $\Si$ and $X$, let $ \gamma_{h,k}(\Si,X) =
\rho_h(\Si,X) \cap \delta_k(\Si,X)$, and let  $\Gamma(h,k) :=
\{\gamma_{h,k}(\Si,X))\}_{\Si,X}$. The following proposition can
be proved simply by combining the arguments used in the previous
two examples.

\begin{proposition}\label{pr:GDef VUT} For all $h,k\geq 0$, $\Gamma(h,k)$ is a consistent  system of congruences, and $GDef_{h,k}$ is the quasi-principal VUT defined by it, and hence  $GDef$ is also a VUT.\ep
\end{proposition}

\subsection{Locally testable unranked tree languages}

For any $k\geq 2$, $\Si$ and $X$, we define the set $\fork_k(t)$
of \emph{$k$-forks} of a $\SX$-tree $t$ thus:
\begin{itemize}
  \item[(1)] if $\hg(t) < k-1$, then $\fork_k(t) = \es$;
  \item[(2)] if $\hg(t) \geq k-1$ and $t = f(t_1,\ldots,t_m)$, then $\fork_k(t) =  \{\rt_k(t)\} \cup \fork_k(t_1)\cup\ldots \cup \fork_k(t_m)$.
\end{itemize}
Clearly, $\fork_k(t)$ is a finite set of $\SX$-trees of height
$k-1$. For example, if $t = f(x,f(y))$, then $\fork_2(t) =
\{f(x,f),f(y)\}$, $\fork_3(t) = \{t\}$ and $\fork_k(t) = \es$ for
all $k\geq 4$. Note that the set of all possible
$k$-forks of $\SX$-trees is infinite.

Now, let $\lambda_k(\Si,X)$ be the relation on $\SXt$ such that for any $s,t\in \SXt$,
$$
  s \, \lambda_k(\Si,X) \,t \; :\LRa\; \st_{k-1}(s) =
  \st_{k-1}(t), \rt_{k-1}(s) = \rt_{k-1}(t), \fork_k(s) = \fork_k(t).
$$
It is easy to see that $\lambda_k(\Si,X)\in \Con(\SXta)$. An unranked $\SX$-tree language is said
to be \emph{$k$-testable} if it is saturated by
$\lambda_k(\Si,X)$, and it is called \emph{locally testable} if it
is $k$-testable for some $k\geq 2$. Let $Loc_k(\Si,X)$ be the set
of all recognizable $k$-testable $\SX$-tree languages, and let
$Loc(\Si,X) := \bigcup_{k\geq 2}Loc_k(\Si,X)$ be the set of all recognizable
locally testable $\SX$-tree languages.

Note that for any $\SX$-tree $t$ of height $\geq k-1$, $\st_{k-1}(t)$ consists of the subtrees of $t$ of height $\leq k-2$, $\rt_{k-1}(t)$ is its root segment of height $k-2$, and $\fork_k(t)$ consists of its forks of height $k-1$. In particular, if $t$ is a string represented as a unary tree, then $\st_{k-1}(t)$ consists of the prefixes of $t$ of length $\leq k-1$, $\rt_{k-1}(t)$ is the suffix of $t$ of length $k-1$, and $\fork_k(t)$ is the set of its substrings of length $k$. Hence, our unranked $k$-testable tree languages are obtained by a natural adaptation of the usual definition of $k$-testable string languages (cf. \cite{Eil76}, for example).

To show that the families
$Loc_k := \{Loc_k(\Si,X)\}$ ($k\geq 2$) and $Loc :=
\{Loc(\Si,X)\}$ are varieties, we consider the systems of
congruences $\Lambda(k) := \{\lambda_k(\Si,X)\}_{\Si,X}$ ($k\geq
2$). For proving the consistency of these systems, we need the
following lemma.

\begin{lemma}\label{le:g-morphism and forks} If $(\iota,\vp) : \SXta \ra \OYta$ is a g-morphism and  $k\geq 2$, then $$\fork_k(t\vp) \:=\: \bigcup\{\rt_k(u\vp)\mid u\in \fork_k(t)\}\cup \bigcup\{\fork_k(s\vp)\mid s\in \st_{k-1}(t)\},$$
for every $t\in \SXt$.
\end{lemma}

\begin{proof} In the course of this proof, we use a couple of times the obvious fact that if $s$ is a subtree of $s'$, then $\fork_k(s) \se \fork_k(s')$ for all $k\geq 2$. Let $RHS$ denote the righthand side of the claimed equality. We proceed by induction on the height of $t\in \SXt$.

If $\hg(t) < k-1$, then $\fork_k(t) = \es$. Moreover, $\fork_k(s\vp) \se \fork_k(t\vp)$ for every $s \in \st_{k-1}(t)$ as $s\in \st(t)$ clearly implies $s\vp\in \st(t\vp)$. Hence, $RHS \se \fork_k(t\vp)$. On the other hand, $\fork_k(t\vp) \se RHS$ because now $t\in \st_{k-1}(t)$.

Let $\hg(t) \geq k-1$ and let $t = f(t_1,\ldots,t_m)$, and assume that the equality holds for all trees of lesser height. Then $t\vp = \iota(f)(t_1\vp,\ldots,t_m\vp)$.

To prove the inclusion $\fork_k(t\vp) \se RHS$, consider any $v\in \fork_k(t\vp)$. Since $\fork_k(t\vp) =  \{\rt_k(t\vp)\} \cup \fork_k(t_1\vp)\cup\ldots \cup \fork_k(t_m\vp)$, there are two possibilities. If $v = \rt_k(t\vp)$, then $v = \rt_k(\rt_k(t)\vp)$ by Lemma \ref{le:g-morphisms and roots}, and hence $v \in RHS$ as $\rt_k(t) \in \fork_k(t)$. If $v\in \fork_k(t_i\vp)$ for some $i\in [m]$, then by the induction assumption, either $v = \rt_k(u\vp)$ for some $u\in \fork_k(t_i) (\se \fork_k(t))$ or $v \in \fork_k(s\vp)$ for some $s\in \st_{k-1}(t_i) (\se \st_{k-1}(t))$. In either case, $v\in RHS$.

Assume now that $v\in RHS$. If $v = \rt_k(u\vp)$ for some  $u \in \fork_k(t)$, we have two cases to consider: $u =\rt_k(t)$ or $u \in \fork_k(t_i)$ for some $i\in [m]$. In the first case, $v = \rt_k(\rt_k(t)\vp) = \rt_k(t\vp)  \in \fork_k(t\vp)$ by Lemma \ref{le:g-morphisms and roots}.  In the second case, $v\in \fork_k(t_i\vp) (\se \fork_k(t\vp))$ by the inductive hypothesis.

Finally, if $v\in \fork_k(s\vp)$ for some $s\in \st_{k-1}(t)$, then $s\in \st_{k-1}(t_i)$ for some $i\in [m]$, and hence $v \in \fork_k(t_i\vp) \se \fork_k(t\vp)$ by the induction assumption. This completes the proof of the inclusion $RHS \se \fork_k(t\vp)$.
\end{proof}

\begin{proposition}\label{LocTestVUT} For every $k\geq 2$, the system of congruences $\Lambda(k)$ is consistent, and $Loc_k$ is the quasi-principal VUT defined by it, and hence also $Loc$ is a VUT.
\end{proposition}

\begin{proof} Consider any  $k\geq 2$ and any g-morphism $(\iota,\vp) :
\SXta \ra \OYta$, and let $s,t\in \SXt$ be such that
$s\,\lambda_k(\Si,X)\, t$. To prove the consistency of
$\Lambda(k)$, we should show that $s\vp\,\lambda_k(\Om,Y)\, t\vp$.

By the proofs of Propositions \ref{pr:RDef_k VUT} and
\ref{pr:Def_k VUT}, we know that $\st_{k-1}(s\vp) =
\st_{k-1}(t\vp)$ and  $\rt_{k-1}(s\vp) = \rt_{k-1}(t\vp)$ follow from $\st_{k-1}(s) = \st_{k-1}(t)$ and $\rt_{k-1}(s) = \rt_{k-1}(t)$,
respectively. Similarly, $\fork_k(s) = \fork_k(t)$ and $\st_{k-1}(s) = \st_{k-1}(t)$ imply $\fork_k(s\vp) = \fork_k(t\vp)$ by Lemma \ref{le:g-morphism and forks}. Hence $s\vp\,\lambda_k(\Om,Y)\, t\vp$.

That $Loc_k$ is the  quasi-principal VUT defined by $\Lambda(k)$
follows immediately from its definition. Finally, $Loc$ is a VUT as the union of the chain $Loc_2 \se Loc_3 \se \ldots$.
\end{proof}

\subsection{Aperiodic tree languages}

To show that the natural unranked counterparts of the aperiodic tree languages \cite{Thom84} form a variety is as easy as in the ranked case \cite{Stei98}.

For any $p,q\in \SXc$ and $t\in \SXt$, let $p\cdot q := q(p)$ and
$t\cdot p := p(t)$. Obviously, $(\SXc,\cdot,\xi)$ is a monoid and
the powers $p^n \,(n\geq 0)$ of a $\SX$-context $p$ are defined as
usual. An unranked tree language $T \se \SXt$ is called
\emph{aperiodic} (or \emph{noncounting}) if there exists a number
$n\geq 0$ such that for all $q,r\in \SXc$ and $t\in \SXt$,
$$t\cdot q^{n+1}\cdot r \in T \: \LRa \: t\cdot q^n\cdot r \in T.$$
If $T$ is aperiodic, the least $n$ for which the above condition
holds, is denoted by $\ia(T)$. Let $Ap(\Si,X)$ be the set of all
recognizable aperiodic $\SX$-tree languages. Let $Ap := \{Ap(\Si,X)\}$.

\begin{proposition}\label{pr:Ap} $Ap$ is a VUT.
\end{proposition}

\begin{proof} It is obvious that $Ap$ satisfies conditions (V1)--(V3). To verify
(V4), consider any $T\in Ap(\Si,X)$ and $p\in \SXc$. If $\ia(T) =
n$, then for all  $q,r\in \SXc$ and $t\in \SXt$,
$$t\cdot q^{n+1}\cdot r \in p^{-1}(T) \, \LRa \, t\cdot q^{n+1}\cdot
(r\cdot p) \in T \, \LRa t\cdot q^n \cdot (r\cdot p) \in T \, \LRa
\, t\cdot q^n \cdot r \in p^{-1}(T),$$ which shows that $p^{-1}(T)
\in Ap(\Si,X)$.

Finally, let $(\iota,\vp) : \SXta \ra \OYta$ be a g-morphism,
$T\in Ap(\Om,Y)$ and $\ia(T) = n$. Define $\hat{\vp} : \SXc \ra
\OYc$ as follows:
\begin{itemize}
  \item[(1)] $\xi\hat{\vp} := \xi$;
  \item[(2)] if $p = f(t_1,\ldots,q,\ldots,t_m)$ for some
  $f\in \Si$, $m\geq 1$, $t_1,\ldots,t_m\in \SXt$ and $q\in \SXc$,
  then $p\hat{\vp} := \iota(f)(t_1\vp,\ldots,q\hat{\vp},\ldots,t_m\vp)$.
\end{itemize}
It is easy to see that $\hat{\vp}$ is a monoid morphism and that
$(t\cdot p)\vp = t\vp\cdot p\hat{\vp}$ for all $t\in \SXt$ and
$p\in \SXc$. This implies that, for all $q,r\in \SXc$ and $t\in
\SXt$,
$$t\cdot q^{n+1}\cdot r \in T\vpi \LRa t\vp\cdot
(q\hat{\vp})^{n+1}\cdot r\hat{\vp} \in T \LRa t\vp\cdot
(q\hat{\vp})^n \cdot r\hat{\vp} \in T \LRa t\cdot q^n\cdot r
\in T\vpi,$$ which shows that $Ap$ satisfies (V5), too.
\end{proof}

\subsection{Piecewise testable tree languages}

As our final example, we consider piecewise testable unranked tree
languages. As shown in \cite{Piir04}, a natural definition of
piecewise subtrees can be based on the well-known homeomorphic
embedding order of trees (cf.~\cite{BaNi98}, for example), and a
corresponding order implicitly underlies the definition of the
piecewise testable `forests' (i.e., finite sequences of unranked
trees) considered in \cite{BoSS12}. The following presentation
parallels that of \cite{Piir04} with the small modifications
introduced in \cite{Stei12}.

For any $\Si$, $X$ and $k\geq 0$, the \emph{homeomorphic embedding
order} $\unlhd$ on $\SXt$ is defined by stipulating that for any
$s,t\in \SXt$, $s \unlhd t$ if and only if
\begin{itemize}
  \item[(1)] $s = t$, or
  \item[(2)] $s = f(s_1,\ldots,s_m)$ and $t = f(t_1,\ldots,t_m)$
  where $s_1 \unlhd t_1,\ldots,s_m \unlhd t_m$, or
  \item[(3)] $t = f(t_1,\ldots,t_m)$ and $s \unlhd t_i$ for some $i\in [m]$.
\end{itemize}

Consider any $\Si$, $X$ and $k\geq 0$. For any $t\in \SXt$, let
$P_k(t) := \{s\in \SXt \mid s \unlhd t, \hg(s) < k\}$, and then
define $\tau_k(\Si,X) := \{(s,t) \mid s,t \in \SXt, P_k(s) =
P_k(t)\}$. Now, an unranked $\SX$-tree language is said to be
\emph{piecewise $k$-testable} if it is saturated by
$\tau_k(\Si,X)$, and it is \emph{piecewise testable} if it is
piecewise $k$-testable for some $k\geq 0$. Let $Pwt_k(\Si,X)$
denote the set of all recognizable piecewise $k$-testable unranked
$\SX$-tree languages, and let $Pwt(\Si,X) := \bigcup_{k\geq
0}Pwt_k(\Si,X)$ be the set of all recognizable piecewise testable
unranked $\SX$-tree languages. We want to prove that the families
$Pwt_k := \{Pwt_k(\Si,X)\}$ ($k\geq 0$) and $Pwt :=
\{Pwt(\Si,X)\}$ are VUTs. For this it suffices to show that for
every $k\geq 0$, $\rmT(k) := \{\tau_k(\Si,X)\}_{\Si,X}$ is a
consistent system of congruences that defines $Pwt_k$.

It is easy to see that $\tau_k(\Si,X)\in \Con(\SXta)$ for every
$k\geq 0$. (Note, however, that $\tau_k(\Si,X)$ is not of finite
index when $k\geq 2$.) For showing that the system $\rmT(k)$
is consistent, we need the following lemma.

\begin{lemma}\label{le:Pw subtrees and g-morphism}
Let $(\iota,\vp) : \SXta \ra \OYta$ be a g-morphism. For any
$k\geq 0$, $s\in \SXt$ and $t \in P_k(s\vp)$, there exists an
$s'\in P_k(s)$ such that $t\in P_k(s'\vp)$.
\end{lemma}

\begin{proof} The proof goes by induction on $k$. The case $k=0$ is trivial
since $P_0(s\vp)$ is empty. If $t\in P_1(s\vp)$, then $t \in \Om
\cup Y$, and now we proceed by induction on $s$. If $s\in \Si\cup
X$, we may let $s'$ be $s$. If $s = f(s_1,\ldots,s_m)$, then $s\vp
= \iota(f)(s_1\vp,\ldots,s_m\vp)$ and we must have $t\in
P_1(s_i\vp)$ for some $i\in [m]$. By our tacit inductive
assumption, there exists an $s'\in P_1(s_i) \se P_1(s)$ such that
$t\in P_1(s'\vp)$.

Assume now that $k\geq 2$ and that the lemma holds for all smaller
values of $k$. Again, we proceed by induction on $s$. If $s\in
\Si\cup X$, then $\hg(s) < k$, and we may set $s' := s$. Let $s =
f(s_1,\ldots,s_m)$ and suppose the claim holds for all smaller
trees. Since  $s\vp = \iota(f)(s_1\vp,\ldots,s_m\vp)$, there are
two cases to consider. If $t\in P_k(s_i\vp)$ for some $i\in [m]$,
the required $s'$ can be found as a piecewise subtree of $s_i$.
Otherwise, $t = \iota(f)(t_1,\ldots,t_m)$ for some $t_1\in
P_{k-1}(s_1\vp),\ldots,P_{k-1}(s_m\vp)$. By the main inductive
assumption, there are trees $s_1'\in P_{k-1}(s_1),\ldots,s_m'\in
P_{k-1}(s_m)$ such that $t_1\in P_{k-1}(s_1'\vp),\ldots,t_m\in
P_{k-1}(s_m'\vp)$. Then $t\in P_k(s'\vp)$ for $s' :=
f(s_1',\ldots,s_m') \in P_k(s)$.
\end{proof}

\begin{proposition}\label{pr:Pwt VUT} For every $k\geq 0$, the
system of congruences $\rmT(k)$ is consistent, and $Pwt_k$
is the quasi-principal VUT defined by it, and hence also $Pwt$ is
a VUT.
\end{proposition}

\begin{proof} That the system $\rmT(k)$ is consistent follows directly
from Lemma \ref{le:Pw subtrees and g-morphism}. Indeed, if
$(\iota,\vp) : \SXta \ra \OYta$ is a g-morphism and $s\,
\tau_k(\Si,X)\, t$, then $s\vp \, \tau_k(\Om,Y)\, t\vp$ because
$P_k(s) = P_k(t)$ implies $P_k(s\vp) = P_k(t\vp)$ by that lemma.

By the definition of $\tau_k(\Si,X)$, a
recognizable $\SX$-tree language $T$ is piecewise $k$-testable if
and only if $T$ is saturated by $\tau_k(\Si,X)$, and this is the
case exactly in case $\tau_k(\Si,X) \se \theta_T$. This means that
$\calV_{\rmT(k)} = Pwt_k$.
\end{proof}

\section{Concluding remarks}\label{se:Concluding remarks}

We have introduced and studied varieties of unranked tree languages that contain languages over all operator and leaf alphabets. We also defined the basic algebraic notions, such as subalgebras, morphisms and direct products, for unranked algebras in a way that allows us to consider algebras over any operator alphabets together. In particular, we have considered regular algebras, i.~e.,  finite unranked algebras in which the operations are defined by regular languages. A bijective correspondence between varieties of unranked tree languages and varieties of regular algebras was established via  syntactic algebras. We have also shown that the natural unranked counterparts of several known varieties of ranked tree languages form varieties in our sense. In many of these examples we made use of a general scheme by which so-called quasi-principal varieties are obtained from certain systems of congruences of term algebras.

Of course, much remains to be done. In particular, many of the example varieties considered here would deserve a deeper study. For example, it is natural to ask for characterizations of the corresponding varieties of regular algebras, or whether there are logics matching some of these varieties.

\newpage

\section*{Appendix: Some proofs}

This appendix contains several proofs that we have either omitted or just outlined in the main text. Most of them are straightforward, technical and rather uninstructive. In many cases, they can be obtained by obvious modifications from earlier similar proofs.

\renewcommand{\theapplemma}{\ref{le:g-morphisms}}
\begin{applemma} Let $\SalgA$, $\OalgB$ and
$\GalgC$ be unranked algebras, and  $(\iota,\vp) : \calA \ra \calB$ and
$(\vk,\psi) : \calB \ra \calC$ be g-morphisms.
\begin{itemize}
  \item[{\rm (a)}] The product $(\iota\vk,\vp\psi) : \calA \ra
  \calC$ is also a g-morphism. Moreover, if $(\iota,\vp)$ and
  $(\vk,\psi)$ are g-epi-, g-mono- or g-isomorphisms, then so is
  $(\iota\vk,\vp\psi)$.
  \item[{\rm (b)}] If $R$ is a g-subalgebra of $\calB$, then $R\vp^{-1}$ is a g-subalgebra of $\calA$. In particular, if $R$ is a $\Psi$-subalgebra of $\calB$ for some $\Psi \se \Om$, then $R\vp^{-1}$ is a $\iota^{-1}(\Psi)$-subalgebra of $\calA$.
  \item[{\rm (c)}] If $S$ is a g-subalgebra of $\calA$, then $S\vp$ is a g-subalgebra of $\calB$. In particular, if $S$ is a $\Psi$-subalgebra of $\calA$ for some $\Psi\se \Si$, then $S\vp$ is a $\iota(\Psi)$-subalgebra of $\calB$.
\end{itemize}
\end{applemma}

\begin{proof} (a) For any $f \in \Si$ and $w \in A^*$,
$$
f_{\calA}(w)\vp\psi = \iota(f)_{\calB} (w\vp_*)\psi =
(\iota\vk)(f)_{\calC} ((w\vp_*)\psi_*)) =
(\iota\vk)(f)_{\calC}(w(\vp\psi)_*).
$$
Thus $(\iota\kappa,\vp\psi)$ is a g-morphism from $\calA$ to
$\calC$. Moreover, if $\iota$, $\kappa$, $\vp$, $\psi$ are
injective, surjective or bijective, so are also $\iota\vk$ and
$\vp\psi$, respectively.

(b) Let $R$ be a $\Psi$-subalgebra of $\calB$ for some $\Psi \se
\Om$. To show that $R\vp^{-1}$ is a $\iota^{-1}(\Psi)$-closed
subset of $\calA$, consider  any $f \in \iota^{-1}(\Psi)$, $m\geq
0$ and $a_1,\ldots,a_m \in R\vp^{-1}$. Since $R$ is $\Psi$-closed,
$\iota(f) \in \Psi$ and $a_1\vp,\dots, a_m\vp \in R$, we get
$$
    f_{\calA} (a_1,\dots,a_m)\vp =
    \iota(f)_{\calB} (a_1\vp,\dots,a_m\vp) \in R,
$$
and hence $f_{\calA} (a_1,\dots,a_m) \in R\vp^{-1}$.

(c) Let $S$ be a $\Psi$-subalgebra of $\calA$ for some $\Psi \se
\Si$. To see that $S\vp$ is a $\iota(\Psi)$-closed subset of
$\calB$, consider any $g \in \iota(\Psi)$, $m \ge 0$, and $b_1$,
\dots, $b_m \in S\vp$. Then $g = \iota(f)$ for some $f \in \Psi$,
and $b_1 = a_1\vp$, \dots, $b_m = a_m\vp$ for some $a_1$, \dots,
$a_m \in S$, and hence
$$
    g_{\calB}(b_1, \dots,b_m)
    = \iota(f)_{\calB} (a_1\vp,\dots,a_m\vp)
    = f_{\calA} (a_1,\dots,a_m)\vp \in S\vp
$$
since $S$ is $\Psi$-closed.
\end{proof}

\renewcommand{\theapplemma}{\ref{le:Quotients&Homomorphisms}}
\begin{applemma} Let $\SalgA$
and $\OalgB$ be any algebras.
\begin{enumerate}
  \item[{\rm (a)}] For any g-congruence $(\si,\theta)$ of $\calA$,
   the natural maps $\theta_\natural : A \ra A/\theta , a \mapsto
   a\theta,$ and $\si_\natural : \Sigma \ra \Sigma/\si, f \mapsto
   f\si,$ define a g-epimorphism $(\si_\natural,\theta_\natural) :
   \calA \ra \calA/(\si,\theta)$.
  \item[{\rm (b)}] The \emph{kernel} $\ker (\iota,\vp) := (\ker
  \iota, \ker \vp)$ of any g-morphism $(\iota,\vp) : \calA \ra
  \calB$  is a g-congruence on $\calA$. If
  $(\iota,\varphi)$ is a g-epimorphism, then $\calA/\ker \,(\iota,
  \varphi) \cong_g \calB$.
\end{enumerate}
\end{applemma}

\begin{proof} (a) If $f \in \Sigma$, $m \ge 0$ and $a_1$, \dots, $a_m \in
A$, then
\begin{align*}
    f_{\calA}(a_1,\dots,a_m)\theta_\natural
    &= f_{\calA}(a_1,\dots,a_m)\theta
    = (f\si)_{\calA/(\si,\theta)} (a_1\theta, \dots, a_m\theta)
    \\
    &= \si_\natural(f)_{\calA/(\si,\theta)} (a_1\theta_\natural,
    \dots, a_m\theta_\natural),
\end{align*}
which proves $(\si_\natural,\theta_\natural)$ to be a g-morphism.
Of course, $\theta_\natural$ and $\si_\natural$ are surjective.

(b) Naturally $\ker \iota \in \Eq(\Si)$ and $\ker \vp \in \Eq(A)$,
and for any $f,g \in \Si$, $m \ge 0$ and $a_1,\dots,a_m,
b_1,\dots,b_m \in A$,
\begin{align*}
    (f,g) \in\ker \iota,  &(a_1, b_1),\dots,(a_m,b_m) \in \ker \vp \\
    &\Ra \iota(f) = \iota(g), a_1\vp = b_1\vp, \dots, a_m\vp
    = b_m \vp \\
    &\Ra
      \iota(f)_{\calB}(a_1\vp, \dots, a_m\vp) =
      \iota(g)_{\calB}(b_1\vp, \dots, b_m\vp) \\
    &\Ra
      f_{\calA}(a_1, \dots, a_m)\vp =
      g_{\calA}(b_1, \dots, b_m)\vp  \\
    &\Ra
      (f_{\calA}(a_1,\dots,a_m),g_{\calA}(b_1,\dots,b_m))\in \ker\vp.
\end{align*}
Thus $\ker(\iota, \vp)$ is a g-congruence.

Let $(\iota,\varphi)$ now be a g-epimorphism.  We show
that $(\vk,\psi): \calA/\ker(\iota,\vp) \to \calB$ is a
g-isomorphism when $\vk : \Si/\ker \iota \ra \Om$ and $\psi :
A/\ker \vp \ra B$ are defined by $\vk : f\ker\iota \mapsto
\iota(f)$ and $\psi : a\ker\vp \mapsto a\vp$, respectively. First
of all, the two mappings are well-defined and injective: for any
$f_1,f_2 \in \Si$,
$$
  \vk(f_1\ker\iota) = \vk(f_2\ker\iota) \;  \LRa \; \iota(f_1)
  = \iota(f_2) \; \LRa \; f_1\ker\iota = f_2\ker\iota,
$$
and for any $a_1,a_2 \in A$,
$$
  (a_1\ker\vp)\psi = (a_2\ker\vp)\psi \; \LRa \; a_1\vp = a_2\vp \;
    \LRa \; a_1 \ker\vp = a_2\ker\vp.
$$
Moreover, the maps are surjective because $\iota$ and $\vp$ are
surjective. Finally, for any $f \in \Si$, $m \ge 0$,
$a_1,\dots,a_m \in A$,
\begin{align*}
    (f\ker\iota&)_{\calA/\ker(\iota,\vp)} (a_1\ker\vp, \dots,
                                         a_m\ker\vp)\, \psi \\
    &= f_{\calA} (a_1,\dots,a_m) (\ker\vp)\; \psi
     = f_{\calA} (a_1, \dots, a_m) \vp \\
    &= \iota(f)_{\calB} (a_1\vp, \dots, a_m\vp)
     = \vk(f\ker\iota)_{\calB} ((a_1\ker\vp)\psi, \dots,
    (a_m\ker\vp)\psi).
\end{align*}
Altogether, $(\vk,\psi)$ is a g-isomorphism between
$\calA/\ker(\iota,\vp)$ and $\calB$.
\end{proof}

\renewcommand{\theappropo}{\ref{pr:TermAlgebraFree}}
\begin{appropo} For any $\Si$ and any $X$,
the term algebra $\SXta$ is \emph{freely generated}
by $X$ over the class of all unranked algebras, that is to say,
\begin{itemize}
  \item[{\rm (1)}]  $\angX = \SXt$, and
  \item[{\rm (2)}] if $\OalgA$ is any unranked algebra, then for
  any pair of mappings $\iota : \Si \ra \Om$ and $\alpha : X \ra A$,
  there is a unique g-morphism
  $(\iota,(\iota,\alpha)_\calA) : \SXta \ra \calA$ such that
  $(\iota, \alpha)_\calA \big|_X  = \alpha$.
\end{itemize}
\end{appropo}

\begin{proof} (1) As $\angX$ is the intersection of those subalgebras of
$\SXta$ that contain $X$, it is clear that $X\se \angX \se \SXt$,
and that to prove $\SXt \se \angX$, it suffices to show that $\SXt
\se B$ for any $\Si$-closed subset $B$ of $\SXta$ for which $X \se
B$. This can be done by tree induction as follows.
\begin{enumerate}
  \item The inclusion $X\se B$ holds by the choice of $B$. For any
  $f \in \Si$, $f = f_{\SXta}(\ve) \in B$ because $B$ is $\Si$-closed.
  Hence also $\Si \se B$ holds.
  \item Let $t = f(t_1,\dots,t_m)$ for some $f \in \Si$, $m > 0$, and
  $t_1$, \dots, $t_m \in \SXt$ such that $t_1$, \dots, $t_m \in B$.
  Then also $t = f(t_1,\dots,t_m) = f_{\SXta}(t_1, \dots, t_m) \in B$.
\end{enumerate}

(2) For any given $\iota:\Sigma\to\Omega$ and $\alpha:X \to A$, a
g-morphism $(\iota,\vp):\SXta \to \calA$ such that $\vp\big|_X =
\alpha$ must satisfy the following conditions:
\begin{itemize}
    \item[(1)] For $x \in X$, $x\vp = \alpha(x)$.
    \item[(2)] For $f \in \Si$, $f\vp = f_{\SXta}(\ve)\vp =
          \iota(f)_{\calA}(\ve)$.
    \item[(3)] For $t= f(t_1,\dots,t_m)$, where $f \in \Si$,
    $m >0$, $t_1$, \dots, $t_m \in \SXt$,
    $$t\vp = f_{\SXta}(t_1,\dots,t_m)\vp
    = \iota(f)_{\calA}(t_1\vp, \dots, t_m\vp).$$ Assuming
    inductively that the values $t_i\vp$ are uniquely defined, this
    defines  a unique value for $t\vp$.
\end{itemize}
It is clear that the thus defined $(\iota,\vp)$ is a g-morphism.
\end{proof}

\renewcommand{\theapplemma}{\ref{le:SHPcommut}}
\begin{applemma}
\begin{itemize}
    \item[{\rm (a)}]  $S_gS_g = S_gS = SS_g = S_g$.
    \item[{\rm (b)}]  $H_gH_g = H_gH = HH_g = H_g$.
    \item[{\rm (c)}]  $P_gP_g = P_gP_f = P_fP_g = P_g$.
    \item[{\rm (d)}]  $S_gH \leq S_gH_g \leq HS_g \leq H_gS_g$.
    \item[{\rm (e)}]  $P_gS \leq P_gS_g \leq SP_g \leq S_gP_g$.
    \item[{\rm (f)}]  $P_gH \leq P_gH_g \leq HP_g \leq H_gP_g$
\end{itemize}
\end{applemma}

\begin{proof}  Statements (a) and (b) hold because obviously $S_gS_g =
S_g$ and $H_gH_g = H_g$.

To prove (c), it clearly suffices to show that $P_gP_g \leq P_g$.
To reduce the notational complexity, we consider g-products with
two factors only. In what follows, $\bfK$ is always any given class of unranked
algebras.

Let $\calA_i = (A_i,\Sigma_i) \in \bfK$ for $i \in [4]$ and let
$\tau_1: \Omega_1 \to \Sigma_1 \times \Sigma_2$ and $\tau_2:
\Omega_2 \to \Sigma_3 \times \Sigma_4$ be mappings, and let
$\calB_1 = (B_1,\Omega_1)$ and $\calB_2 = (B_2,\Omega_2)$ be
algebras that are isomorphic to $\tau_1(\calA_1,\calA_2)$ and
$\tau_2(\calA_3,\calA_4)$ via the respective isomorphisms $\vp_1:
A_1 \times A_2 \to B_1$ and $\vp_2: A_3 \times A_4 \to B_2$. Next,
define a mapping $\lambda: \Gamma \to \Omega_1 \times \Omega_2$.
Then any algebra  $\calC = (C,\Gamma)$ isomorphic to
$\lambda(\calB_1,\calB_2)$ is a typical member of
$P_g(P_g(\bfK))$. We should show that $\calC \in P_g(\bfK)$.

Define $\mu: \Gamma \to \Sigma_1 \times \Sigma_2 \times \Sigma_3
\times \Sigma_4$ so that $\mu(g) = (f_1, f_2,f_3,f_4)$ if
$\lambda(g) = (h_1,h_2)$, $\tau_1(h_1) = (f_1,f_2)$, and
$\tau_2(h_2) = (f_3,f_4)$. Then
$\mu(\calA_1,\calA_2,\calA_3,\calA_4)$ is a g-product of members
of $\bfK$. We show that
$$\psi : \mu(\calA_1,\calA_2,\calA_3,\calA_4) \ra \lambda(\calB_1,\calB_2), \, (a_1,a_2,a_3,a_4) \mapsto ((a_1,a_2)\vp_1, (a_3,a_4)\vp_2),$$
is an isomorphism. Clearly, $\psi$ is a bijection since $\vp_1$
and $\vp_2$ are bijections. Consider any $g \in \Gamma$, $m \ge
0$, and $(a_{11}, a_{12}, a_{13}, a_{14})$, \dots, $(a_{m1},
a_{m2}, a_{m3}, a_{m4}) \in A_1 \times A_2 \times A_3 \times A_4$.
If $\lambda(g) = (h_1,h_2)$, $\tau_1(h_1) = (f_1,f_2)$, and
$\tau_2(h_2) = (f_3,f_4)$, then
\begin{align*}
    g&_{\mu(\calA_1,\calA_2,\calA_3,\calA_4)}
      ((a_{11}, a_{12}, a_{13}, a_{14}), \dots,
       (a_{m1}, a_{m2}, a_{m3}, a_{m4}))\psi\\
   &=
   ( (f_1)_{\calA_1} (a_{11}, \dots, a_{m1}),
     (f_2)_{\calA_2} (a_{12}, \dots, a_{m2}), \\
   &\phantom{mm}
     (f_3)_{\calA_3} (a_{13}, \dots, a_{m3}),
     (f_4)_{\calA_4} (a_{14}, \dots, a_{m4}) ) \psi \\
   &=
   (( (f_1)_{\calA_1} (a_{11}, \dots, a_{m1}),
             (f_2)_{\calA_2} (a_{12}, \dots, a_{m2}) ) \vp_1 , \\
   &\phantom{mm}
     ( (f_3)_{\calA_3} (a_{13}, \dots, a_{m3}),
             (f_4)_{\calA_4} (a_{14}, \dots, a_{m4}) )\vp_2  ) \\
   &=
   (( (h_1)_{\tau_1(\calA_1,\calA_2)}
             ( (a_{11}, a_{12}), \dots, (a_{m1}, a_{m2}) ) ) \vp_1, \\
   &\phantom{mm}
    ( (h_2)_{\tau_2(\calA_3,\calA_4)}
             ( (a_{13}, a_{14}), \dots, (a_{m3}, a_{m4}) ) ) \vp_2 ) \\
   &=
   ( (h_1)_{\calB_1}
             ( (a_{11},a_{12})\vp_1, \dots, (a_{m1},a_{m2})\vp_1 ), \\
   &\phantom{mm}
     (h_2)_{\calB_2}
             ( (a_{13},a_{14})\vp_2, \dots, (a_{m3},a_{m4})\vp_2
             ) ) \\
   &= g_{\lambda(\calB_1,\calB_2)}
           (  ( (a_{11},a_{12})\vp_1, (a_{13},a_{14})\vp_2 ),
             \dots,
              ( (a_{m1},a_{m2})\vp_1, (a_{m3},a_{m4})\vp_2 ) ) \\
   &= g_{\lambda(\calB_1,\calB_2)}
           (  (a_{11},a_{12}, a_{13},a_{14})\psi,
             \dots,
              (a_{m1},a_{m2}, a_{m3},a_{m4})\psi ).
\end{align*}
So $\psi$ is also a morphism, and $\calC \in P_g(\bfK)$ as $\calC
\cong \lambda(\calB_1, \calB_2) \cong
\mu(\calA_1,\calA_2,\calA_3,\calA_4)$.

In each of (d), (e) and (f), it suffices to prove the second
inequality because the first and the third inequalities are
obvious.

To complete the proof of (d), we should show that $S_gH_g \leq
HS_g$. To construct a typical member  $\GalgC$ of $S_gH_g(\bfK)$,
let $\calA = (A,\Sigma)$ be in $\bfK$, $(\iota,\vp): \calA \ra
\calA'$ be a g-epimorphism, $\OalgB$ be a g-subalgebra of
$\calA'$, and let $(\vk,\psi): \calB \to \calC$ be a
g-isomorphism. Now $\calB\vp^{-1} =
(B\vp^{-1},\iota^{-1}(\Omega))$ is a g-subalgebra of $\calA$. If
we choose a subset $\Sigma'$ of $\iota^{-1}(\Omega)$ so that the
restriction of $\iota$ to $\Sigma'$ is a bijection $\iota':
\Sigma' \to \Omega$, then $\calD = (B\vp^{-1},\Sigma')$ is a
g-subalgebra of $\calA$.

Next we define a $\Gamma$-algebra $\calE = (B\vp^{-1},\Gamma)$ so
that for each $g \in \Gamma$, $g_{\calE} = f_{\calD}$ for the $f
\in \Sigma'$ such that $g = \vk(\iota'(f))$. Then $(\iota'\vk,
1_{B\vp^{-1}}) : \calD \ra \calE$ is a g-isomorphism. Indeed, if
$f \in \Sigma'$ and $w \in (B\vp^{-1})^*$, then
$$
    f_{\calD} (w) 1_{B\vp^{-1}}  = f_{\calD} (w)
    = \vk(\iota'(f))_{\calE} (w 1_{B\vp^{-1}}).
$$
This means that $\calE \in S_g(\bfK)$. Next, we show that
$\vp\psi: \calE \to \calC$ is an epimorphism. Clearly,
$B\vp^{-1}\vp\psi = C$. Consider any $g \in \Gamma$ and $w \in
(B\vp^{-1})^*$. Let $f \in \Sigma'$ and $h \in \Omega$ be such
that $\iota'(f) = h$ and $\vk(h) = g$. Then
$$
    g_{\calE} (w) \vp\psi = f_{\calD} (w) \vp \psi
    = f_{\calA} (w) \vp \psi = h_{\calA'} (w\vp) \psi
    = h_{\calB} (w\vp) \psi = g_{\calC} (w\vp\psi).
$$
Thus $\calC \in HS_g(\bfK)$.

The proof of (e) is complete when we show that $P_gS_g \leq SP_g$.
Any algebra $\calD = (D,\Ga)$ in $P_gS_g(\bfK)$ is isomorphic to a
g-product $\lambda(\calC_1, \dots \calC_n)$, where $n\geq 0$, and
for each $i \in [n]$, $\calA_i = (A_i,\Sigma_i)$ is a member of
$\bfK$,  $\calB_i = (B_i,\Omega_i)$ is a g-subalgebra of
$\calA_i$, $\calC_i = (C_i,\Gamma_i)$ is an algebra g-isomorphic
to $\calB_i$ via some g-isomorphism $(\iota_i,\vp_i): \calB_i \to
\calC_i$, and  $\lambda: \Gamma \to \Gamma_1 \times \dots \times
\Gamma_n$ is a mapping. It suffices to show that $\lambda(\calC_1,
\dots \calC_n)\in SP_g(\bfK)$.

To do this, we define the mapping $\vk: \Gamma \to \Sigma_1 \times
\dots \times \Sigma_n$ so that $\vk(g) = (f_1, \dots, f_n)$ if
$\lambda(g) = (g_1, \dots, g_n)$ and $\iota_1(f_1) = g_1$, \dots,
$\iota_n(f_n) = g_n$. Then $\vk(\Gamma) \subseteq \Omega_1 \times
\dots \times \Omega_n$. Now $\vk(\calA_1, \dots, \calA_n)$ is in
$P_g(\bfK)$ and $\vk(\calB_1, \dots, \calB_n)$ is a subalgebra of
$\vk(\calA_1, \dots, \calA_n)$. Indeed, for any $g\in \Gamma$,  $m
\ge 0$ and $(b_{11}, \dots, b_{1n})$, \dots, $(b_{m1}, \dots,
b_{mn}) \in B_1 \times \dots \times B_n$, if $\vk(g) = (f_1,
\dots, f_n)$, then
\begin{align*}
    g_{\vk(\calA_1, \dots, \calA_n)} &((b_{11}, \dots, b_{1n}),
    \dots, (b_{m1}, \dots, b_{mn}))      \\
    &=
    ((f_1)_{\calA_1} (b_{11}, \dots, b_{m1}), \dots,
     (f_n)_{\calA_n} (b_{1n}, \dots, b_{mn}))
    \in B_1 \times \dots \times B_n
\end{align*}
since each $\calB_i$ is an $\Omega_i$-closed subset of $A_i$. Next
we verify that the mapping
$$\vp: B_1 \times \dots \times B_n \to C_1
\times \dots \times C_n, \, (b_1, \dots, b_n)\vp \mapsto
(b_1\vp_1, \dots, b_n\vp_n),$$ defines an isomorphism from
$\vk(\calB_1, \dots, \calB_n)$ to $\lambda(\calC_1, \dots,
\calC_n)$. Clearly, $\vp$ is bijective. Moreover, for $g$, the
$b_{ij}$'s and $f_i$'s as above,
\begin{align*}
    g_{\vk(\calB_1, \dots, \calB_n)}& ((b_{11}, \dots, b_{1n}),
    \dots, (b_{m1}, \dots, b_{mn}))\vp\\
    &= ((f_1)_{\calB_1} (b_{11}, \dots, b_{m1}), \dots,
     (f_n)_{\calB_n} (b_{1n}, \dots, b_{mn})) \vp \\
    &=
    ((f_1)_{\calB_1} (b_{11}, \dots, b_{m1}) \vp_1, \dots,
     (f_n)_{\calB_n} (b_{1n}, \dots, b_{mn}) \vp_n) \\
    &=
    (\iota_1(f_1)_{\calC_1} (b_{11}\vp_1, \dots, b_{m1}\vp_1), \dots,
     \iota_n(f_n)_{\calC_n} (b_{1n}\vp_n, \dots, b_{mn}\vp_n)) \\
    &=
    g_{\lambda(\calC_1, \dots, \calC_n)}
      ((b_{11}\vp_1, \dots, b_{1n}\vp_n), \dots,
       (b_{m1}\vp_1, \dots, b_{mn}\vp_n)) \\
    &=
    g_{\lambda(\calC_1, \dots, \calC_n)}
      ((b_{11}, \dots, b_{1n})\vp, \dots,
       (b_{m1}, \dots, b_{mn})\vp).
\end{align*}
This shows that $\lambda(\calC_1, \dots \calC_n)\in SP_g(\bfK)$.

Finally, we prove the critical inequality $P_gH_g \leq HP_g$ of
(f).  Let $n\geq 0$ and for each $i\in [n]$, let $\calA_i =
(A_i,\Sigma_i)$ be an algebra  in $\bfK$ and let $(\iota_i,\vp_i):
\calA_i \to \calB_i$ be a g-epimorphism from $\calA_i$ onto an
algebra $\calB_i = (B_i,\Omega_i)$. Let $\lambda: \Gamma \to
\Omega_1 \times \dots \times \Omega_n$ be a mapping and $\psi:
\lambda(\calB_1, \dots, \calB_n) \to \calC$ an isomorphism from
$\lambda(\calB_1, \dots, \calB_n)$ to some $\calC = (C,\Gamma)$.
Then $\calC$ is a typical member of $P_gH_g(\bfK)$.

Define $\vk: \Gamma \to \Sigma_1 \times \dots \times \Sigma_n$ as
follows. If $g \in \Gamma$ and $\lambda(g) = (h_1, \dots, h_n)$,
choose any $f_1 \in \Sigma_1$, \dots, $f_n \in \Sigma_n$ for which
$\iota_1(f_1) = h_1$, \dots, $\iota_n(f_n) = h_n$ and set $\vk(g)
= (f_1, \dots, f_n)$. Then the mapping
$$\vp: A_1
\times \dots \times A_n \to C, \, (a_1, \dots, a_n)\vp \mapsto
(a_1 \vp_1, \dots, a_n\vp_n)\psi,$$ is an epimorphism from
$\vk(\calA_1, \dots,\calA_n)$ to $\calC$. Firstly, $\vp$ is
surjective since $\psi$ and $\vp_1$, \dots, $\vp_n$ are
surjective. Secondly, for any $g \in \Gamma$, $m \ge 0$, and
$(a_{11}, \dots, a_{1n})$, \dots, $(a_{m1}, \dots, a_{mn}) \in A_1
\times \dots \times A_n$, if $\vk(g) = (f_1, \dots, f_n)$, then
\begin{align*}
    g&_{\vk(\calA_1, \dots, \calA_n)}
    ( (a_{11}, \dots, a_{1n}), \dots, (a_{m1}, \dots, a_{mn}))
    \vp \\
    &=
    ((f_1)_{\calA_1} (a_{11}, \dots, a_{m1}), \dots,
     (f_n)_{\calA_n} (a_{1n}, \dots, a_{mn})) \vp \\
    &=
    ((f_1)_{\calA_1} (a_{11}, \dots, a_{m1}) \vp_1, \dots,
     (f_n)_{\calA_n} (a_{1n}, \dots, a_{mn}) \vp_n) \psi \\
    &=
    ( \iota_1(f_1)_{\calB_1} (a_{11}\vp_1, \dots, a_{m1}\vp_1), \dots,
      \iota_n(f_n)_{\calB_n} (a_{1n}\vp_n, \dots, a_{mn}\vp_n)) \psi \\
    &=
    g_{\lambda(\calB_1, \dots, \calB_n)}
    ( (a_{11}\vp_1, \dots, a_{1n}\vp_n), \dots,
      (a_{m1}\vp_1, \dots, a_{mn}\vp_n)) \psi \\
    &=
    g_{\calC}
    ( (a_{11}\vp_1, \dots, a_{1n}\vp_n) \psi, \dots,
      (a_{m1}\vp_1, \dots, a_{mn}\vp_n) \psi ) \\
    &=
    g_{\calC}
    ( (a_{11}, \dots, a_{1n}) \vp, \dots,
      (a_{m1}, \dots, a_{mn}) \vp ).
\end{align*}
This shows that $\calC \in HP_g(\bfK)$.
\end{proof}

\renewcommand{\theapplemma}{\ref{le:SyntCongrSat}}
\begin{applemma} For any subset $H\se A$ of an unranked algebra $\SalgA$, $\theta_H$ is the greatest congruence of $\calA$ that saturates $H$.
\end{applemma}

\begin{proof} It is clear that $\theta_H$ is an equivalence on $A$. The
congruence property follows immediately by Lemma \ref{le:Transl
and Congr}: if $a \, \theta_H \, b$ and $q\in \Tr(\calA)$, then
for any $p\in \Tr(\calA)$,
$$p(q(a))\in H \, \LRa \, p(q)(a)\in H \, \LRa \, p(q)(b) \in H \, \LRa \, p(q(b))\in H,$$
which shows that $q(a) \, \theta_H \, q(b)$. The congruence
$\theta_H$ also saturates $H$. Indeed, if $a \, \theta_H \, b$ and
$a\in H$, then also $b\in H$ because $a = 1_A(a)$ and $b =
1_A(b)$.

Finally, let $\theta$ be a congruence of $\calA$ saturating $H$,
and consider any $a,b\in A$ such that $a \, \theta \, b$. Again by
Lemma \ref{le:Transl and Congr},  $p(a) \, \theta \, p(b)$ for
every $p\in \Tr(\calA)$. Since $\theta$ saturates $H$, this means
that, for every $p\in \Tr(\calA)$, $p(a)\in H$ iff $p(b) \in H$,
i.e., that $a \, \theta_H \, b$.
\end{proof}

\renewcommand{\theappropo}{\ref{pr:ClosureProp of Rec}}
\begin{appropo}
(a)  $\es$, $\SXt \in \RecSX$ and $\RecSX$ is closed under all Boolean operations.
\end{appropo}

\begin{proof} Clearly, $\es$ and $\SXt$ are recognized even by trivial
$\Si$-algebras. If $T,U \in \RecSX$, then there are regular
$\Si$-algebras $\SalgA$ and $\SalgB$ such that $T = F\vpi$ and $U
= G\psii$ for some morphisms $\vp : \SXta \ra \calA$, $\psi :
\SXta \ra \calB$ and some $F\se A$, $G\se B$. Then $\SXt\setminus
T = (A\setminus F)\vpi \in \RecSX$ and $T\cap U = (F\times
G)\etai$ for the morphism $\eta : \SXta \ra \calA \times \calB, t
\mapsto (t\vp,t\psi)$. Since $\calA\times\calB$ is regular, this
means that also $T\cap U \in \RecSX$.
\end{proof}


\begin{thebibliography}{20}

\bibitem{Alme90} J.~Almeida, On pseudovarieties, varieties of
languages, filters of congruences, pseudo-identities and related
topics, \emph{Algebra Universalis} 27 (1990) 333--350.

\bibitem{BaNi98} F.~Baader, T.~Nipkow, {\em Term Rewriting
and All That}, Cambridge University Press, Cambridge, 1998.

\bibitem{Berg12} C.~Bergman, \emph{Universal Algebra. Fundamentals
and Selected Topics}, CRC Press, Taylor \& Francis Group, Boca
Raton, Fl, 2012.

\bibitem{BoSS12} M.~Boja\'{n}czyk, L.~Segoufin, H.~Straubing,
Piecewise testable tree languages, \emph{Log. Methods in
Comput. Sci.} 8 (2012) 1--32.

\bibitem{BoWa08} M.~Boja\'{n}czyk, I.~Walukiewicz, Forest algebras, in: J.~Flum, E.~Gr\"adel, T.~Wilke (Eds.), \emph{Automata and Logic: History and Perspectives}, Texts in Logic and Games, Vol. 2, Amsterdam University Press, Amsterdam 2008, pp. 107--132.


\bibitem{BrMW01} A.~Br\"uggemann-Klein, M.~Murata, D.~Wood,
Regular tree and hedge languages over unranked alphabets, Version
1, HKUST Theoretical Computer Science
Center Research Report HKUST-TCS-2001-05, Hong Kong, 2001.

\bibitem{BuSa81} S.~Burris, H.P.~Sankappanavar,
{\em A Course in Universal Algebra}, Springer-Verlag, New York,
1981.


\bibitem{TATA07}
H.~Comon, M.~Dauchet, R.~Gilleron, C.~L\"oding, F.~Jacquemard,
D.~Lugiez, S.~Tison, M.~Tommasi, \emph{Tree Automata Techniques
and Applications}, available online:
http://www.grappa.univ-lille3.fr/tata, 2007.


\bibitem{CrLT05} J.~Cristau, C.~L\"oding, W.~Thomas,
Deterministic automata on unranked trees, in: M.~Li\'{s}kiewicz,
R.~Reischuk (Eds.), \emph{Fundamentals of Computation Theory}, Proc. 15th Internat. Symp. FCT 2005, L\"ubeck 2005, Lect. Notes Comput. Sci. 3623, Springer, Berlin, 2005, pp.~68--79.

\bibitem{DaPr02} B.A.~Davey, H.A.~Priestley, \emph{Introduction to
Lattices and Order}, second ed., Cambridge University Press, Cambrige
2002.

\bibitem{Eil76} S.~Eilenberg, \emph{Automata, Languages, and Machines. Vol. B}, Academic Press, New York, 1976.



\bibitem{GeSt84} F.~G\'ecseg, M.~Steinby, {\em Tree Automata},
Akad\'emiai Kiad\'o, Buda\-pest, 1984.

\bibitem{GeSt97} F.~G\'ecseg, M.~Steinby, Tree languages, in:
G.~Rozenberg, A.~Salomaa (Eds.), {\em Handbook of Formal Languages},
Vol. 3, Springer, Berlin, 1997, pp.~1--69.


\bibitem{GrSc90} E.~Graczy\'nska, D.~Schweigert, Hypervarieties
of a given type, \emph{Algebra Universalis} 27 (1990)
303--318.


\bibitem{JuPoTh94} E.~Jurvanen, A.~Potthoff, W.~Thomas,
Tree languages recognizable by regular frontier check, in:
G.~Rozenberg, A.~Salomaa (Eds.), \emph{Developments in Language
Theory}, World Scientific, Singapore, 1994, pp.~3--17.



\bibitem{MaNi07} W.~Martens, J.~Niehren, On the mininization of
XML Schemas and tree automata for unranked trees, \emph{J.
Comput. Syst. Sci.} 73 (2007) 550--583.


\bibitem{Neve02} F.~Neven, Automata, Logic, and XML,
in: J.~Bradfield (Ed.), \emph{Computer Science Logic},
Proc. 16th Internat. Workshop, CSL 2002, Edinburgh,
UK, 2002, Lect. Notes Comput. Sci. 2471,
Springer, Berlin, 2002, pp.~2--26.



\bibitem{PaQu68} C.~Pair, A.~Quere, D\'efinition et etude
bilangages r\'eguliers, \emph{Inform. Control} 13(6)
(1968) 565--593.

\bibitem{Piir04} V.~Piirainen, Piecewise testable tree languages,
TUCS Technical Report 634, Turku Centre for Computer Science,
Turku, 2004.

\bibitem{Schw93} D.~Schweigert, Hyperidentities,
in: I.G.~Rosenberg, G.~Sabidussi (Eds.), \emph{Algebras
and Orders}, Kluwer Academic Publishers, Dordrecht, 1993,
pp.~405--506.

\bibitem{Schw07} T.~Schwentick, Automata for XML -- A survey,
\emph{J. Comput. Syst. Sci.} 73 (2007) 289--315.

\bibitem{SeVi02} L.~Segoufin, V.~Vianu, Validating
streaming XML documents, in: L.~Popa (Ed.),
\emph{ACM SIGMOD-SIGACT-SIGART Symposium
on Principles of Database Systems}, Proc. 21th Internat. Symp.,
PODS 2002, Wisconsin, USA, 2002, ACM, 2002, pp.~53--64.

\bibitem{Stei79} M.~Steinby, Syntactic algebras and varieties of
recognizable sets, in: M.C.~Gaudel, J.P.~Jouannaud (Eds.), \emph{Les Arbres en Alg\`ebre et en
Programmation}, Proc. 4th CAAP, Lille 1979, University of Lille,
Lille, 1979, pp.~226--240.

\bibitem{Stei92} M.~Steinby, A theory of tree language varieties,
in: M.~Nivat, A.~Podelski, (Eds.), \emph{Tree Automata and Languages}, North-Holland, Amsterdam, 1992, pp.~57--81.

\bibitem{Stei98} M.~Steinby, General varieties of tree languages,
\emph{Theor. Comput. Sci.} 205 (1998) 1--43.

\bibitem{Stei05} M.~Steinby, Algebraic classifications of regular
tree languages, in: V.B. Kudryavtsev, I.G.~Rosenberg (Eds.), \emph{Structural Theory of Automata, Semigroups, and Universal Algebra}, Springer, Dordrecht, 2005, pp.~381--432.

\bibitem{Stei12} M.~Steinby: On the solidity of general varieties
of tree languages, \emph{Discuss. Math. Gen. Algebra Appl.} 32
 (2012) 23--53.

\bibitem{Taka75} M.~Takahashi, Generalizations of regular sets and
their application to a study of context-free languages,
\emph{Inform. Control} 27 (1975) 1--36.

\bibitem{That67} J.W.~Thatcher, Characterizing Derivation Trees
of Context-Free Grammars through a Generalization of Finite
Automata Theory, \emph{J. Comput. Syst. Sci.} 1(4)
(1967) 317--322.

\bibitem{That73} J.W.~Thatcher, Tree automata: an informal survey,
in: A.V.~Aho (Ed.), \emph{Currents in the Theory of Computing},
Prentice-Hall, Englewood Cliffs, NJ, 1973, pp.~143--172.

\bibitem{Thom84} W.~Thomas, Logical aspects in the study of tree
languages, in: B.~Courcelle (Ed.), \emph{9th Colloquium on Trees in Algebra and Programming}, Proc. 9th CAAP, Bordeaux, 1984,
Cambridge University Press, London, 1984, pp.~31--49.


\end{thebibliography}
\end{document}